\documentclass[letterpaper,journal]{IEEEtran}
\usepackage[T1]{fontenc}
\usepackage[cmintegrals]{newtxmath}
\usepackage{newtxtext}

\AtBeginDocument{
    \setlength{\columnsep}{0.25in}
}

\IEEEoverridecommandlockouts

\usepackage{amsmath}
\usepackage{subeqnarray}
\usepackage{cases}
\usepackage{cite}
\usepackage{graphicx}
\usepackage{tabularx,booktabs}
\usepackage{diagbox}
\usepackage[ruled,linesnumbered]{algorithm2e}
\usepackage{algorithmic}
\usepackage{esint}
\usepackage{cuted}
\usepackage{stfloats}
\usepackage{multirow}
\usepackage{ntheorem}
\usepackage{threeparttable}
\usepackage{bbm}
\usepackage{tabu}
\usepackage{bbding}
\usepackage{tikz}
\usepackage{xcolor}
\usepackage{makecell}
\usepackage{subfigure}
\usepackage{url}

\theoremseparator{:}
\newtheorem{proposition}{\textcolor{black}{Proposition}}
\newtheorem{lemma}{Lemma}
\newtheorem{theorem}{Theorem}

\theoremstyle{nonumberplain}
\theorembodyfont{\normalfont}
\newtheorem{proof}{Proof}

\providecommand{\qed}{\unskip\nobreak\hfill\ensuremath{\square}}

\newcommand{\posgt}[1]{\mathbf{Y}_{#1}}
\newcommand{\posest}[1]{\hat{\mathbf{Y}}_{#1}}

\definecolor{newextractedpurple}{RGB}{127,0,127}
\newcommand{\FullCov}{\scalebox{1.0}{\textmd{\Checkmark}}}
\newcommand{\PartCov}{\(\scalebox{1.0}{$\triangle$}\)}
\newcommand{\NoCov}{\(\scalebox{1.0}{$\times$}\)}

\def\BibTeX{{\rm B\kern-.05em{\sc i\kern-.025em b}\kern-.08em
    T\kern-.1667em\lower.7ex\hbox{E}\kern-.125emX}}

\begin{document}
\bstctlcite{IEEEexample:BSTcontrol}

\setlength{\columnsep}{0.25in}

\title{Edge-Assisted Multi-View Localization\\ for Low-Altitude Economy under\\ GPS-Challenged Environments}

\author{Zhengru~Fang$^{\dagger}$, 
        Huanhuan~Lou$^{\dagger}$,
        Senkang~Hu,
        Yihang~Tao, 
        Zongdian~Li,~\IEEEmembership{Member,~IEEE},
        Yiqin~Deng,~\IEEEmembership{Member,~IEEE},
        Jingjing~Wang,~\IEEEmembership{Senior Member,~IEEE},
        and~Yuguang~Fang,~\IEEEmembership{Fellow,~IEEE}
\thanks{$^{\dagger}$Z. Fang and H. Lou contributed equally to this work (co-first authors).}
\thanks{Z. Fang was with the JC STEM Lab of Smart City and the Department of Computer Science, City University of Hong Kong, Hong Kong, when this work was conducted. He is currently with the Department of Electronic and Computer Engineering, The Hong Kong University of Science and Technology, Hong Kong (e-mail: eezfang@ust.hk).}
\thanks{H. Lou, S. Hu, Y. Tao, and Y. Fang are with the JC STEM Lab of Smart City and the Department of Computer Science, City University of Hong Kong, Hong Kong (e-mail: \{huanhulou2, senkanghu2-c, yihangtao2-c\}@my.cityu.edu.hk; my.fang@cityu.edu.hk).}
\thanks{Z. Li is with Zhejiang University, Hangzhou 310027, China (e-mail: zongdianli@zju.edu.cn).}
\thanks{Y. Deng is with the School of Data Science, Lingnan University, Tuen Mun, Hong Kong, China (e-mail: yiqindeng@ln.edu.hk).}
\thanks{J. Wang is with the School of Cyber Science and Technology, Beihang University, Beijing, China (e-mail: drwangjj@buaa.edu.cn).}
\thanks{A preliminary version of this work was presented at the IEEE Global Communications Conference (GLOBECOM), 2025 \cite{fang2025ovib}.}
}

\maketitle

\begin{abstract}
Unmanned aerial vehicles (UAVs) serving the low-altitude economy require reliable localization in urban canyons, indoor facilities, and other GPS-challenged environments. Visual matching with a geo-tagged database provides an alternative source for absolute positioning, but onboard computation and energy limits motivate offloading the database and matching pipeline to an edge server. The resulting localization quality depends on what visual information can reach the edge in time under varying wireless-communication and edge-computing resources. In this paper, we propose a network-adaptive edge-assisted multi-view localization framework that combines scalable orthogonality-regularized variational information bottleneck (O-VIB) encoding, value-of-information (VOI)-guided request control, and value-aware edge scheduling. We design an O-VIB model that supports nested latent prefixes from 8 to 128 dimensions and four UAV view modes. Each UAV requests edge assistance when the predicted localization-risk reduction exceeds the communication and service costs. On CARLA multi-view UAV data, VOI-guided control can lower the mean and 95th-percentile (p95) route errors by 24.8\% and 31.0\%, respectively, relative to budgeted periodic offloading under a matched per-route traffic budget. In indoor UAV experiments with motion-capture ground truth, our design can lower the mean position error by 28.0\% relative to uncompressed all-view CLIP retrieval while cutting the descriptor traffic by 98.6\%, using a 0.145 KB semantic representation. Under high congestion, a VOI-weighted scheduler with waiting-age and deadline shaping can lower the edge-side p95 latency of the top-10\% high-value requests from 137.7 ms to 32.8 ms.
\end{abstract}


\section{Introduction}
Unmanned aerial vehicles (UAVs) are an important platform for inspection and other applications in the low-altitude economy~\cite{11449237}. Their operation in urban canyons, industrial facilities, campuses, and indoor spaces requires reliable localization under \emph{GPS-challenged environments}, where GPS positioning is weak, intermittent, or spoofed. In such environments, vision-based matching with a geo-tagged database provides an alternative source for absolute positioning. Recent visual-inertial simultaneous localization and mapping (SLAM), digital-twin positioning, and visual place recognition systems support localization when GPS positioning is unreliable \cite{Zhang2022CVIDS,Gao2024DigitalTwin,Wang2022Pavement,AnyLoc2024,MixVPR2023,EigenPlaces2023}. Robustness can improve further with label-free adaptation of onboard vision models~\cite{ma2026learning} and with multi-camera configurations, which retain a usable view under occlusion, repetitive layouts, or rapid viewpoint changes. Running the complete localization pipeline onboard is nevertheless expensive, because the geo-tagged database and the matching stage often exceed the computation and energy budget on a small UAV. Edge assistance moves these components to a ground edge server, but the transmitted visual information must be reduced to fit the available wireless link while preserving localization quality \cite{Cao2024VisualSLAM,Yuan2024Split,10480247}.

\begin{figure}[t]
  \centering
  \includegraphics[width=\columnwidth]{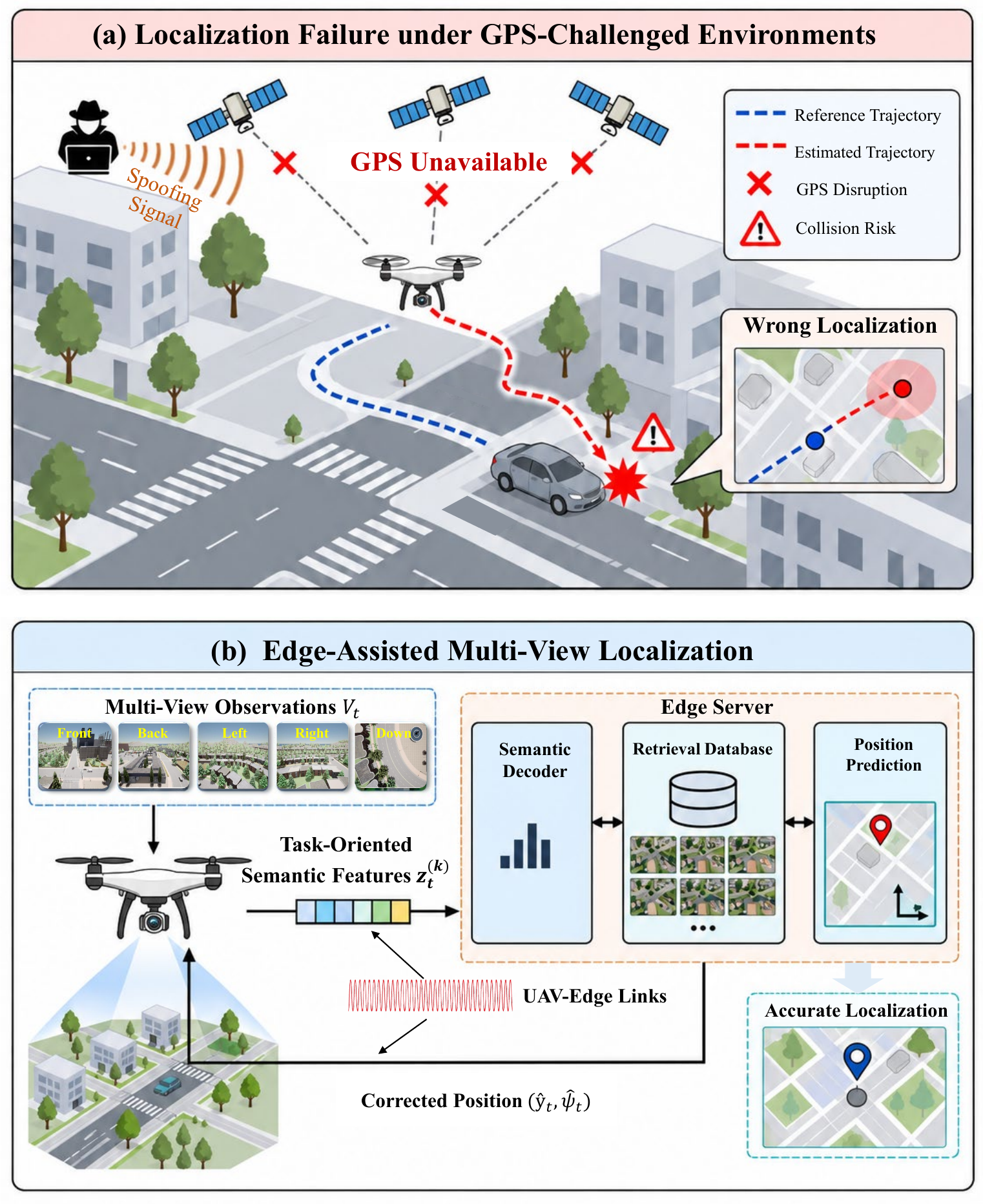}
  \caption{Motivating UAV scenario and edge-assisted multi-view localization pipeline. Panel (a) illustrates localization drift and collision risk when GPS is unavailable or unreliable. Panel (b) shows the UAV selecting a task-oriented multi-view representation, transmitting it to the ground edge server, and receiving a corrected position after decoding and database-assisted localization.}
  \label{fig:scenario}
  \vspace{-2mm}
\end{figure}

Take the UAV scenario in Fig.~\ref{fig:scenario} as an example. When GPS is unavailable or spoofed, local motion estimates drift from the true trajectory and the localization risk grows. The proposed system therefore keeps the fast local loop on the UAV, while an edge request carries only a selected view-rate semantic representation. The ground edge server decodes this representation, searches the geo-tagged database, estimates the UAV's position, and returns a correction. Because multiple UAV clients share the uplink and the edge server, the localization quality depends jointly on semantic payload selection, request timing, and service ordering.

This pipeline therefore creates four coupled challenges. First, \emph{what to transmit}: multi-view visual embeddings may contain substantial cross-view redundancy, while localization depends mainly on location-discriminative and task-relevant features. Second, \emph{how much to transmit}: wireless-link throughput and onboard energy vary over time, so a fixed high-rate representation becomes infeasible under bandwidth bottlenecks, whereas a fixed low-rate representation underutilizes available bandwidth when channel conditions improve. Third, \emph{when to transmit}: local odometry can maintain short-term pose estimates, and an edge request is valuable mainly when uncertainty growth justifies the cost. Fourth, \emph{which request to serve}: many UAV clients may contend for the same uplink and edge GPU, and first-come-first-served or rate-only scheduling serves them without considering the localization error of each request.

Existing approaches address only fragments of this problem. Conventional image and video codecs such as JPEG, H.264, H.265, and WebP optimize perceptual fidelity rather than localization utility, so aggressive compression may discard task-critical geometric cues \cite{wallace1992jpeg,H264,bossen2012hevc,10605825}. In parallel, recent task-oriented and information-bottleneck based semantic communication systems reduce task-irrelevant payloads, yet many of them train a single-rate encoder or assume a periodic upload model, without adapting view mode, semantic rate, and request timing together \cite{Shao2024EdgeVideo,Wei2023FederatedSem,Yuan2024Split,10480247,fang2025ton}. A third line of work on information-aware status updating and edge scheduling highlights the importance of freshness, correctness, and service ordering, but it usually assumes fixed-size state packets or leaves scheduling decoupled from semantic compression \cite{fang2022age,Chen2023AoII,Wang2024AoIURLLC,Feng2024AoIHARQ,Li2023ESMO,SHEPHERD2023,AlpaServe2023}. To the best of our knowledge, no existing study combines state-triggered localization requests, view- and rate-adaptive semantic encoding, and value-aware scheduling under shared wireless-communication and edge-computing resources.

Motivated by these observations, we propose a network-adaptive edge-assisted multi-view localization framework for UAV applications in the low-altitude economy under GPS-challenged environments. An orthogonality-regularized variational information bottleneck (O-VIB) encoder determines what and how much visual information to transmit, while the task-level value of information (VOI) guides when to request a position correction and which request the edge server should serve. Both components serve the same goal of reducing the localization error under limited wireless-communication and edge-computing resources. The main contributions are summarized as follows:
\begin{itemize}
    \setlength{\topsep}{3pt}
    \setlength{\itemsep}{1pt}
    \item We develop a scalable multi-view O-VIB encoder for UAV localization. Through nested rate dropout and view masks, one model supports latent prefixes from 8 to 128 dimensions and four UAV view modes. A variational rate term limits the retained information, while an orthogonality penalty reduces redundant projection directions. At 8 KB/s after feature extraction, the fixed-rate 32-dimensional O-VIB configuration takes 50.4 ms for encoding, decoding, and transmission under the Jetson Orin NX profile, whereas the fastest evaluated image codec takes 2.92 s.
    \item We design a VOI-guided online request-and-mode selection policy. A UAV requests edge assistance only when the predicted reduction in localization risk is worth the costs for the communication, delay requirement, and energy consumption, and it selects the view-rate action with the largest net VOI. Under a matched traffic budget in CARLA, this policy can lower the mean and 95th-percentile (p95) route errors by 24.8\% and 31.0\%, respectively, relative to budgeted periodic offloading.
    \item We formulate multi-UAV edge scheduling as a VOI-weighted drift-plus-penalty (DPP) problem and prove the queue stability and an \(O(1/V)\) utility gap under the assumptions of Theorem~\ref{thm:voi_scheduler}. We also evaluate a shaped scheduler with waiting-age and deadline terms, which can lower the edge-side p95 latency of the top-10\% high-value requests by 76.2\% under high congestion, from 137.7 ms to 32.8 ms.
    \item Our code and a multi-view UAV dataset of 357,690 five-view frames from eight CARLA towns, with RGB, depth, semantic segmentation, and six-degree-of-freedom pose labels, have been made open source for the research community. In indoor UAV experiments with motion-capture ground truth, our design can lower the mean position error by 28.0\% and the descriptor traffic by 98.6\% relative to uncompressed all-view CLIP retrieval. A supplementary unmanned ground vehicle (UGV) experiment shows that the design also applies to ground platforms.
\end{itemize}

\section{Related Work}\label{sec:related}

\subsection{Task-Oriented and Semantic Communication}
Task-oriented, or semantic, communication transmits only the information needed by the receiver's task. Shao \textit{et~al.}~\cite{Shao2024SemanticTheory} provide an information-theoretic framework, while Zhang \textit{et~al.}~\cite{Zhang2025Intellicise} survey open problems. Features learned based on the information bottleneck principle have been used for edge video analytics and federated learning~\cite{Shao2024EdgeVideo,Wei2023FederatedSem,fang2025ton}, and similar designs have been applied to sensing and modeling in the metaverse~\cite{Meng2024CrossSystem,Wang2023SemanticSensing}. Yuan \textit{et~al.}~\cite{Yuan2024Split} and Furutanpey \textit{et~al.}~\cite{10480247} study scalable or lightweight split encoders, while CASVA~\cite{CASVA2022} and ILCAS~\cite{Wu2024ILCAS} adapt video configurations to network conditions. Nested dropout and Matryoshka representation learning order the embedding so that its leading dimensions form shorter representations~\cite{rippel2014nested,kusupati2022mrl}. However, these works usually do not adapt the view mode, the semantic rate, and the request time jointly. Different from these works, we let a VOI-guided controller select them online.

\subsection{Edge-Assisted Visual Localization}
When GPS signals are unavailable, visual localization provides another source for absolute positioning. Collaborative visual-inertial SLAM, digital twins, and coded pavement references support multi-agent, indoor, and road positioning~\cite{Zhang2022CVIDS,Gao2024DigitalTwin,Wang2022Pavement}. Visual place recognition and learned matching improve retrieval under viewpoint and appearance changes~\cite{MixVPR2023,EigenPlaces2023,AnyLoc2024,Izquierdo2024OT,VPRLightGlue2023,VPRGlueStick2023}. Multi-view observations have also been exploited in collaborative perception and scene reconstruction~\cite{tao2026learning,tao2026fruc}. Cao \textit{et~al.}~\cite{Cao2024VisualSLAM} offload collaborative visual SLAM to an edge server to reduce the onboard cost. However, these methods usually assume that complete observations are available at the localizer, and edge-assisted systems generally upload them periodically with fixed compression. In contrast, in this paper, we jointly select the views, the semantic rate, and the request time by the predicted localization value.

\subsection{Information-Aware Status Updating and VOI}
Status-update systems decide whether a new update is worth sending. Age of information measures update staleness, whereas age of incorrect information accounts for the period when the receiver holds an inaccurate estimate. These metrics have guided the transmission design for massive access, UAV, and platooning networks~\cite{fang2022age,Wang2024AoIURLLC,Feng2024AoIHARQ,Abedi2024SAoI,Chen2023AoII}. The value of information further measures the expected reduction in task loss brought by an update~\cite{howard1966information}, and it has been used to trigger the transmission of updates in networked control systems~\cite{Soleymani2022VoI}. Wang \textit{et~al.}~\cite{Wang2023AdaptiveVideo} use deep reinforcement learning to decide when to upload for edge video analytics. However, these works mainly decide when to send an update of largely fixed size and content. In contrast, our VOI estimate determines whether to request edge assistance, which views and rate to send, and how the edge server makes the scheduling decisions.

\subsection{Multi-Client Edge Scheduling for Inference}
A shared edge server determines which request to serve first. Existing edge systems schedule video frames, cached models, and concurrent deep neural network jobs~\cite{Li2023ESMO,Ekya2022,Padmanabhan2023Gemel,Khani2023RECL,Nan2023CloudEdge,FaaSLearner2026,SHEPHERD2023,REEF2022,AlpaServe2023}. However, these schedulers usually do not consider how much a request improves the localization accuracy of its client. In this paper, we combine the queue backlog in Lyapunov optimization~\cite{neely2010sno} with the predicted localization value of each request under wireless-communication and edge-computing budgets. Table~\ref{tab:related_work_comparison} compares the proposed framework with representative works.

\begin{table*}[t]
\centering
\caption{Feature coverage across representative works.}
\label{tab:related_work_comparison}
\footnotesize
\setlength{\tabcolsep}{2.5pt}
\renewcommand{\arraystretch}{1.12}
\begin{tabular*}{\textwidth}{@{\extracolsep{\fill}}l*{9}{c}@{}}
\toprule
\textbf{Feature}
& \cite{Shao2024EdgeVideo}
& \cite{Yuan2024Split}
& \cite{CASVA2022,Wu2024ILCAS}
& \cite{MixVPR2023,EigenPlaces2023,AnyLoc2024}
& \cite{Wang2023AdaptiveVideo}
& \cite{Cao2024VisualSLAM}
& \cite{SHEPHERD2023,REEF2022,AlpaServe2023}
& \cite{fang2025ovib}
& \textbf{Proposed work} \\
\midrule
Task-oriented compression & \FullCov{} & \FullCov{} & \PartCov{} & \NoCov{} & \PartCov{} & \NoCov{} & \NoCov{} & \FullCov{} & \textbf{\FullCov{}} \\
Variable-rate representation & \PartCov{} & \FullCov{} & \FullCov{} & \NoCov{} & \PartCov{} & \NoCov{} & \NoCov{} & \NoCov{} & \textbf{\FullCov{}} \\
Request control & \NoCov{} & \NoCov{} & \PartCov{} & \NoCov{} & \FullCov{} & \PartCov{} & \NoCov{} & \NoCov{} & \textbf{\FullCov{}} \\
View selection & \NoCov{} & \NoCov{} & \PartCov{} & \NoCov{} & \NoCov{} & \NoCov{} & \NoCov{} & \NoCov{} & \textbf{\FullCov{}} \\
Coarse-to-fine retrieval & \NoCov{} & \NoCov{} & \NoCov{} & \PartCov{} & \NoCov{} & \PartCov{} & \NoCov{} & \NoCov{} & \textbf{\FullCov{}} \\
Uncertainty feedback & \NoCov{} & \NoCov{} & \NoCov{} & \NoCov{} & \FullCov{} & \PartCov{} & \NoCov{} & \NoCov{} & \textbf{\FullCov{}} \\
Multi-client scheduling & \NoCov{} & \NoCov{} & \PartCov{} & \NoCov{} & \NoCov{} & \PartCov{} & \FullCov{} & \NoCov{} & \textbf{\FullCov{}} \\
Edge-testbed validation & \NoCov{} & \FullCov{} & \FullCov{} & \NoCov{} & \NoCov{} & \FullCov{} & \FullCov{} & \PartCov{} & \textbf{\FullCov{}}  \\
\bottomrule
\end{tabular*}
\par\vspace{3pt}
\begin{minipage}{\textwidth}
\footnotesize \textbf{Note:} Columns are identified by citation. \FullCov{} denotes fully addressed, \PartCov{} denotes partially addressed or conditionally evaluated, and \NoCov{} denotes not addressed.
\end{minipage}
\vspace{-1em}
\end{table*}
\renewcommand{\arraystretch}{1.0}

\section{System Overview}\label{sec:system}

\subsection{System Architecture}

\begin{figure*}[t]
  \centering
  \includegraphics[width=0.98\textwidth]{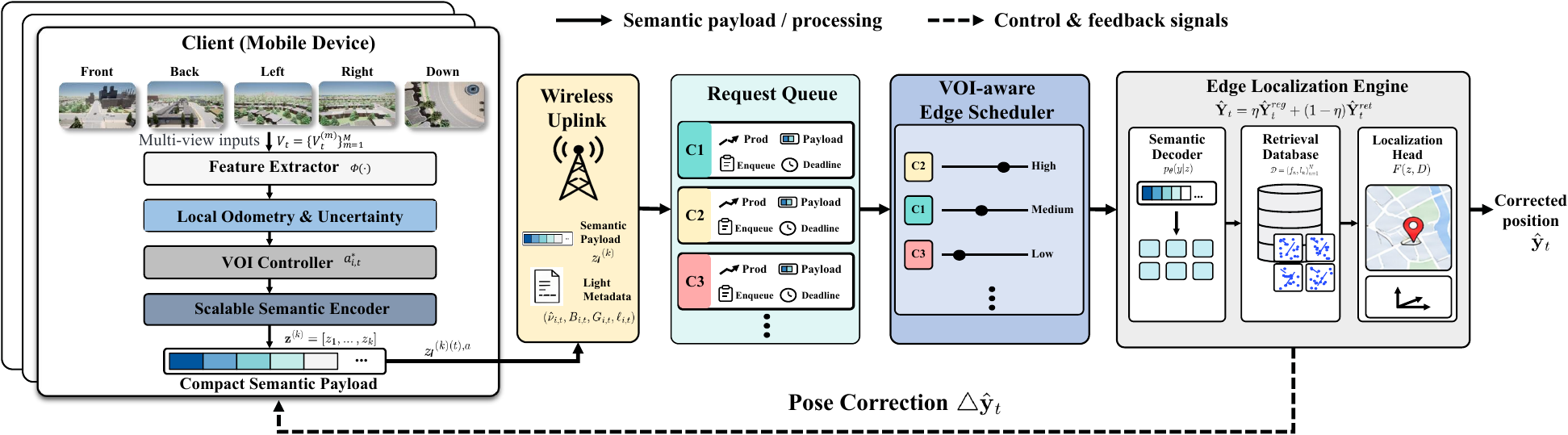}
  \caption{System architecture of network-adaptive edge-assisted UAV localization. Each UAV client maintains a fast local loop, requests an edge correction when its predicted value is high, and receives position feedback after shared wireless and edge-scheduling stages.}
  \label{fig:architecture}
  \vspace{-3mm}
\end{figure*}

We consider a time-slotted system in which \(N\) UAV clients are assisted by a ground edge server in the GPS-challenged environment of Fig.~\ref{fig:scenario}, and the system architecture is shown in Fig.~\ref{fig:architecture}. Each UAV carries \(M=5\) onboard cameras for the front, back, left, right, and downward views. At slot \(t\), the client extracts the feature $\mathbf X_t^{(m)}=\Phi(V_t^{(m)})$ of the image \(V_t^{(m)}\) from view \(m\) with the backbone \(\Phi\) in Section~\ref{sec:feature_extraction}. The features of $\mathbf V_t=\{V_t^{(m)}\}_{m=1}^{M}$ form the multi-view feature $\mathbf X_t$. The edge server stores a geo-tagged database $\mathcal D=\{(f_n,l_n)\}_{n=1}^{N_{\mathcal D}}$, where \(f_n\) is a reference descriptor, \(l_n\) is its location, and \(N_{\mathcal D}\) is the number of entries.

\begin{figure*}[t]
   \centering
   \includegraphics[width=0.98\textwidth]{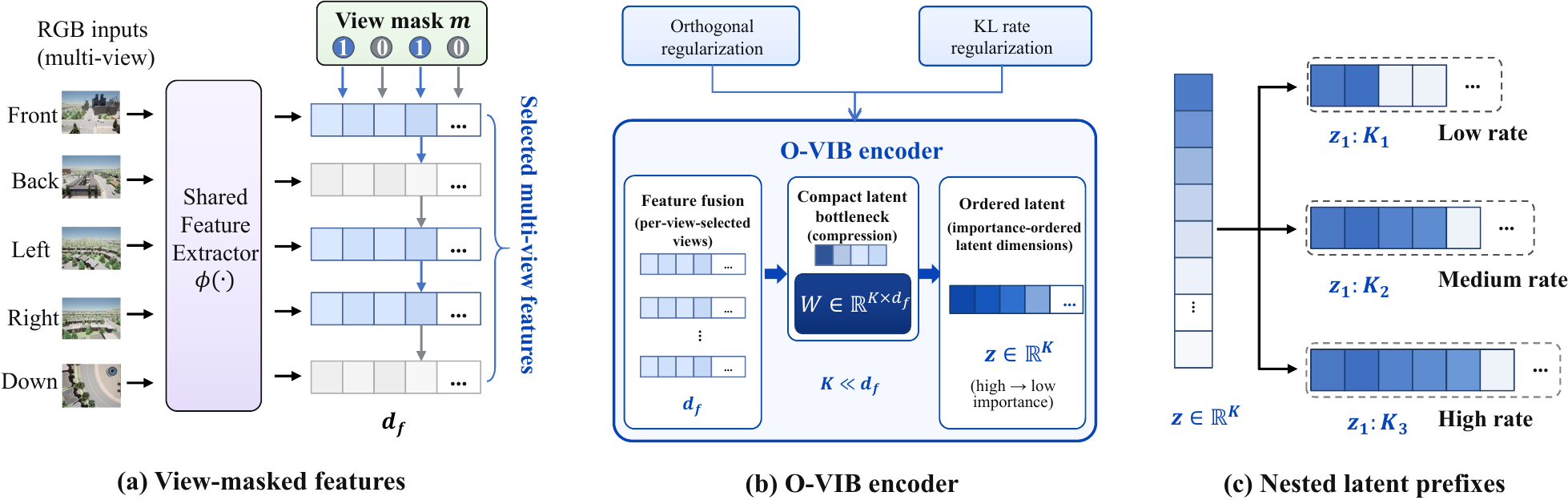}
   \caption{Scalable O-VIB encoder for view-rate adaptive semantic localization. View masks select the available camera subset, while nested latent prefixes provide low-, medium-, and high-rate semantic payloads for one deployed encoder.}
   \label{fig:ovib_encoder}
   \vspace{-3mm}
\end{figure*}

As shown in Fig.~\ref{fig:ovib_encoder}, the client encodes \(\mathbf X_t\) into $\mathbf Z_t=\mathcal E(\mathbf X_t;\phi)$ and transmits the representation to the edge server, where \(\phi\) denotes the encoder parameters. The edge localizer \(\mathcal F\), whose parameters \(\theta\) contain the reconstruction-decoder and localization-head parameters, estimates the UAV client's position from \(\mathbf Z_t\) and \(\mathcal D\). It obtains the position $\posest{t}^{\mathrm{reg}}$ predicted by the regression branch and the location $\posest{t}^{\mathrm{ret}}$ retrieved from \(\mathcal D\) by nearest-neighbor cosine matching of the reconstructed descriptor, and combines them as $\posest{t}=\mathcal F(\mathbf Z_t;\theta,\mathcal D)=\eta\posest{t}^{\mathrm{reg}}+(1-\eta)\posest{t}^{\mathrm{ret}}$, where the mixing weight \(\eta\in[0,1]\) is selected on the validation set. To reduce the cost of exhaustive retrieval, we summarize the database with scene and tile prototypes. A scene prototype averages the descriptors from one scene, and a tile prototype averages the descriptors whose locations fall in one spatial tile. The retrieval branch first matches the scene prototypes, then the tile prototypes, and searches only the selected tiles. This reduces the fine-search cost from $O(N_{\mathcal D}d_{\mathrm r})$ to $O(N_{\mathrm{tile}}d_{\mathrm r})$, excluding the prototype-matching stages, where $d_{\mathrm r}$ is the descriptor dimension and $N_{\mathrm{tile}}$ is the number of descriptors in the selected tiles.

\subsection{Communication Model}

We model the uplink from a UAV client to the ground edge server as a block-fading channel whose gain is constant during the transmission of one request. For analysis purpose, we assume that the achievable rate is $R=B\log_2(1+Pg/(N_0B))$, where \(B\) is the allocated bandwidth, \(P\) is the client transmit power, and \(N_0\) is the noise power spectral density. Following the standard decomposition into path loss, shadowing, and small-scale fading~\cite{goldsmith2005wireless}, the channel gain is $g=g_0(\varrho_0/\varrho)^\kappa10^{\xi/10}|h|^2$. Here, \(g_0\) is the path gain at the reference distance \(\varrho_0\), \(\varrho\) is the distance between the client and the edge server, \(\kappa\) is the path-loss exponent, \(\xi\) is the shadowing gain in dB, and \(h\) is the small-scale fading coefficient. At slot \(t\), a request containing \(\mathbf Z_t\) has \(\mathcal C(\mathbf Z_t)\) transmitted bits, including its metadata, and incurs the uplink delay $\tau_t=\mathcal C(\mathbf Z_t)/R_t$, where \(R_t\) is the achievable rate \(R\) at slot \(t\). In our evaluation, \(R_t\) follows the controlled throughput profiles described in Section~\ref{sec:eval}.

\subsection{Problem Formulation}

Since the localization relies on what the UAV client sends to the edge server on the ground, our objective is to minimize the localization error subject to the communication constraint:
\begin{equation}
\begin{aligned}
\min_{\Theta}\quad &\mathbb{E}\!\bigl[\|\posest{t}-\posgt{t}\|_{2}^{2}\bigr]\\
\text{s.t.}\quad &\tau_t=\frac{\mathcal C(\mathbf Z_t)}{R_t}\le T_{\max},
\end{aligned}
\label{eq:problem}
\end{equation}
where $\posest{t}$ and $\posgt{t}$ are the estimated and true positions of the UAV client, \(\Theta=\{\phi,\theta\}\) contains the encoder and localizer parameters, and \(T_{\max}\) is the time available for the uplink transmission in one localization period. The constraint is equivalent to the rate constraint $\mathcal C(\mathbf Z_t)/T_{\max}\le R_t$. Thus, the visual information must be compressed to the extent that sufficient information for localization can be sent over the limited channel. Lowercase $\mathbf x,\mathbf z,\mathbf y$ denote realizations of the flattened feature \(\mathbf X_t\), latent representation \(\mathbf Z_t\), and position \(\posgt{t}\), respectively. We address the localization problem in three parts. Section~\ref{sec:method} trains one encoder that serves different communication constraints, Section~\ref{sec:voi_control} determines when to send a request and in which view-rate mode, and Section~\ref{sec:voi_scheduler} schedules the requests of multiple UAV clients at the edge server.


\section{Scalable Multi-View O-VIB Encoder}\label{sec:method}

\subsection{Task-Oriented Feature Extraction}\label{sec:feature_extraction}

We use a contrastive language-image pre-training (CLIP) based vision backbone for multi-view feature extraction.
Each image $V_t^{(m)}$ is processed by a shared CLIP ViT-B/32 encoder~\cite{radford2021clip}. After the preprocessing $\pi(\cdot)$, the embedding of view $m$ is $\mathbf{X}_t^{(m)}=\Phi(V_t^{(m)})=f_{\text{CLIP}}(\pi(V_t^{(m)});\theta_\Phi)\in\mathbb R^d$, where $\theta_\Phi$ denotes the pretrained CLIP parameters and $d=512$ is the embedding dimension. We use this pretrained backbone because it provides a reusable visual basis across scenes and avoids training a task-specific backbone for the resource-constrained client. We then normalize each embedding as $\tilde{\mathbf X}_t^{(m)}=\mathbf X_t^{(m)}/\|\mathbf X_t^{(m)}\|_2$ and stack the normalized column vectors as $\mathbf X_t=[\tilde{\mathbf X}_t^{(1)},\ldots,\tilde{\mathbf X}_t^{(M)}]\in\mathbb R^{d\times M}$.

\subsection{Task-Oriented Feature Compression}\label{sec:compression}

The multi-view feature $\mathbf{X}_t \in \mathbb{R}^{d \times M}$ is flattened column-wise into $\mathbf{x}\in\mathbb{R}^{d_f}$ for joint encoding, where $d_f:=Md$. Let \(K\) denote the full latent dimension and \(\mathbf z\in\mathbb R^{K}\) the full encoder output. Since the uplink bandwidth is limited, the client needs a compact and task-relevant representation~\cite{fang2025ton}. The information bottleneck (IB) principle~\cite{tishby1999ib} and its variational form~\cite{alemi2017deep} provide the theoretical framework for learning such a representation. The IB principle seeks a stochastic encoder $q_{\phi}(\mathbf{z}|\mathbf{x})$ that minimizes $I(\mathbf{x};\mathbf{z})$ for compactness while maximizing $I(\mathbf{z};\mathbf{y})$ for localization relevance (localization quality). Thus, its objective is
\begin{equation} \label{OP:IB}
  \min_{\phi}\; 
  \underbrace{\beta\,I(\mathbf x;\mathbf z)}_{\text{Rate}}-
  \underbrace{I(\mathbf z;\mathbf y)}_{\text{Localization quality}},
\end{equation}
where the nonnegative hyperparameter $\beta$ controls the trade-off between compression and localization accuracy.
\subsubsection{Variational Rate Term}

The rate term $I(\mathbf x;\mathbf z)$ is controlled through a prior $p(\mathbf z)$ on the latent representation. The encoder is modeled as a diagonal Gaussian distribution $q_{\phi}(\mathbf z\mid\mathbf x)=\mathcal N\!\bigl(\boldsymbol\mu_{\phi}(\mathbf x),\operatorname{diag}\boldsymbol\sigma^{2}_{\phi}(\mathbf x)\bigr)$, and the latent representation is sampled as $\mathbf z=\boldsymbol\mu+\boldsymbol\sigma\odot\boldsymbol\epsilon$ with $\boldsymbol\epsilon\sim\mathcal N(\mathbf 0,\mathbf I)$ through the reparameterization trick~\cite{kingma2014vae}. Note that $q_\phi(\mathbf z\mid\mathbf x)$ depends on the input $\mathbf x$, whereas the prior is a fixed distribution independent of $\mathbf x$, and we set it to the standard Gaussian $p(\mathbf z)=\mathcal N(\mathbf 0,\mathbf I)$. With this prior, the variational rate surrogate is
\begin{equation}
\begin{aligned}
\mathcal R(\phi)
&:=\mathbb E_{\mathbf x}\bigl[\mathrm{KL}\bigl(q_{\phi}(\mathbf z\mid\mathbf x)\,\Vert\,p(\mathbf z)\bigr)\bigr]\\
&=\tfrac12\,\mathbb E_{\mathbf x}\sum_{i=1}^{K}\bigl(\mu_i^{2}+\sigma_i^{2}-\log\sigma_i^{2}-1\bigr),
\end{aligned}
\label{eq:rate_term}
\end{equation}
where $\mathrm{KL}(\cdot\Vert\cdot)$ denotes the Kullback-Leibler divergence. This prior gives a tractable rate surrogate, and Lemma~\ref{prop:chain_rule} shows that it upper-bounds $I(\mathbf x;\mathbf z)$. For the fixed-rate model in~\cite{fang2025ovib}, we instead use a log-uniform prior $p(z_i)\propto|z_i|^{-1}$, which yields an automatic relevance determination (ARD) penalty that prunes uninformative coordinates~\cite{molchanov2017variational}. Different from the Gaussian prior, this prior is \emph{improper}, i.e., $\int_{\mathbb R}|z_i|^{-1}\,\mathrm dz_i$ diverges and it cannot be normalized to a probability density. Its ARD regularizer is written, up to an additive constant, using the approximation in~\cite{molchanov2017variational} as
\begin{equation}
\begin{aligned}
  \mathcal D_{\textsc{ard}}(\mathbf x)
  := {}& \sum_{i=1}^{K}\!\Bigl[\tfrac{1}{2}\log\!\bigl(1+\rho_i^{-1}\bigr)\Bigr.\\
      &\Bigl.{}-c_{1}\,\mathrm{sigm}\!\bigl(c_{2}+c_{3}\log\rho_i\bigr)+c_{1}\Bigr],
  \label{eq:ardkl}
\end{aligned}
\end{equation}
where $\rho_i:=\sigma_i^{2}/\mu_i^{2}$ is the noise-to-signal ratio of coordinate $i$, $\mathrm{sigm}(u)=1/(1+e^{-u})$ is the sigmoid function, and $(c_{1},c_{2},c_{3})=(0.63576,\,1.87320,\,1.48695)$ are the constants fitted in~\cite{molchanov2017variational}. We write $\mathcal R_{\textsc{ard}}(\phi):=\mathbb E_{\mathbf x}[\mathcal D_{\textsc{ard}}(\mathbf x)]$ for this variant. Both expectations over $\mathbf x$ are estimated by minibatch averaging.

\begin{lemma}
\label{prop:chain_rule}
Let $q_{\phi}(\mathbf z\mid\mathbf x)$ be any encoder, let $p(\mathbf z)$ be any proper prior density, and let $q_{\phi}(\mathbf z)=\int q_{\phi}(\mathbf z\mid\mathbf x)p(\mathbf x)\,\mathrm{d}\mathbf x$ be the induced marginal. The mutual information under the joint density $p(\mathbf x)q_{\phi}(\mathbf z\mid\mathbf x)$ satisfies
\begin{align}
 I(\mathbf x;\mathbf z)
 &=\mathbb E_{\mathbf x}\!
    \Bigl[\mathrm{KL}\bigl(q_{\phi}(\mathbf z\mid\mathbf x)\,\Vert\,p(\mathbf z)\bigr)\Bigr]
   -\mathrm{KL}\bigl(q_{\phi}(\mathbf z)\,\Vert\,p(\mathbf z)\bigr)
   \label{eq:chain_rule_exact}\\
 &\le \mathbb E_{\mathbf x}\!
    \Bigl[\mathrm{KL}\bigl(q_{\phi}(\mathbf z\mid\mathbf x)\,\Vert\,p(\mathbf z)\bigr)\Bigr].
   \label{eq:rate_upper}
\end{align}
\end{lemma}

\begin{proof}
Expanding the mutual information gives $I(\mathbf x;\mathbf z)=\mathbb E_{\mathbf x,\mathbf z}[\log q_{\phi}(\mathbf z\mid\mathbf x)-\log q_{\phi}(\mathbf z)]$. Adding and subtracting $\mathbb E_{\mathbf z}[\log p(\mathbf z)]$ yields \eqref{eq:chain_rule_exact}. The second term in \eqref{eq:chain_rule_exact} is nonnegative for a proper density, and dropping it yields \eqref{eq:rate_upper}.
\qed
\end{proof}

\begin{lemma}\label{prop:izy_bound}
For any decoder $p_{\theta}(\mathbf{y}|\mathbf{z})$ and the joint density $p(\mathbf z,\mathbf y)=\int p(\mathbf x,\mathbf y)q_{\phi}(\mathbf z\mid\mathbf x)\,\mathrm{d}\mathbf x$ induced by the Markov chain $\mathbf y\to\mathbf x\to\mathbf z$, the mutual information between the latent representation $\mathbf{z}$ and the task variable $\mathbf{y}$ is lower-bounded by
\begin{equation}\label{eq:izy_bound}
  I(\mathbf{z};\mathbf{y})
  \ge
  \mathbb{E}_{\mathbf{z},\mathbf{y}}\!\left[\log p_{\theta}(\mathbf{y}|\mathbf{z})\right]
  + \mathrm{h}(\mathbf{y}),
\end{equation}
where $\mathrm{h}(\mathbf{y}) = -\mathbb{E}_{\mathbf{y}}[\log p(\mathbf{y})]$ is the differential entropy of the continuous position $\mathbf{y}$ and is constant with respect to $(\phi,\theta)$.
\end{lemma}

\begin{proof}
By definition,
\begin{equation}
I(\mathbf{z};\mathbf{y}) = \mathrm{h}(\mathbf{y}) - \mathrm{h}(\mathbf{y}|\mathbf{z}).
\end{equation}
Moreover,
\begin{equation}
\mathbb{E}_{q_\phi(\mathbf z)}\!\Bigl[\mathrm{KL}\!\bigl(p(\mathbf y|\mathbf z)\,\|\,p_\theta(\mathbf y|\mathbf z)\bigr)\Bigr] \ge 0,
\end{equation}
which implies
\begin{equation}
\mathbb{E}_{\mathbf{z},\mathbf{y}}\!\left[\log p(\mathbf{y}|\mathbf{z})\right]
\ge
\mathbb{E}_{\mathbf{z},\mathbf{y}}\!\left[\log p_\theta(\mathbf{y}|\mathbf{z})\right].
\end{equation}
Since
\begin{equation}
\mathbb{E}_{\mathbf{z},\mathbf{y}}\!\left[\log p(\mathbf{y}|\mathbf{z})\right]
= -\mathrm{h}(\mathbf{y}|\mathbf{z}),
\end{equation}
substitution into $I(\mathbf z;\mathbf y)=\mathrm h(\mathbf y)-\mathrm h(\mathbf y|\mathbf z)$ yields \eqref{eq:izy_bound}, which is the variational lower bound of~\cite{barber2003im}.
\qed
\end{proof}

\begin{theorem}\label{thm:ib_ard}
For the proper Gaussian prior $p(\mathbf z)=\mathcal N(\mathbf 0,\mathbf I)$ used by the scalable encoder and any decoder $p_{\theta}(\mathbf y\mid\mathbf z)$, the IB objective $\mathcal L_{\mathrm{IB}}:=\beta I(\mathbf x;\mathbf z)-I(\mathbf z;\mathbf y)$ in \eqref{OP:IB} satisfies
\begin{equation}
\mathcal L_{\mathrm{IB}}
\le
\beta\,\mathcal R(\phi)
-\mathbb E_{\mathbf z,\mathbf y}\bigl[\log p_{\theta}(\mathbf y\mid\mathbf z)\bigr]
-\mathrm h(\mathbf y).
\label{eq:vib_upper}
\end{equation}
Since $-\mathrm h(\mathbf y)$ is constant with respect to $(\phi,\theta)$, minimizing the right-hand side of \eqref{eq:vib_upper} is equivalent to solving
\begin{equation}
  \min_{\phi,\theta}\;
  \beta\,\mathcal R(\phi)
  -\mathbb E_{\mathbf z,\mathbf y}
      \bigl[\log p_{\theta}(\mathbf y\mid\mathbf z)\bigr].
  \label{eq:vib_obj}
\end{equation}
\end{theorem}

\begin{proof}
Lemma~\ref{prop:chain_rule} gives $\beta I(\mathbf x;\mathbf z)\le\beta\mathcal R(\phi)$ for $\beta\ge 0$, and Lemma~\ref{prop:izy_bound} gives $-I(\mathbf z;\mathbf y)\le-\mathbb E_{\mathbf z,\mathbf y}[\log p_{\theta}(\mathbf y\mid\mathbf z)]-\mathrm h(\mathbf y)$. Adding the two inequalities yields \eqref{eq:vib_upper}, and removing the constant $-\mathrm h(\mathbf y)$ yields \eqref{eq:vib_obj}.
\qed
\end{proof}

In practice, we optimize \eqref{eq:vib_obj} with the Gaussian prior for the scalable model. The fixed-rate model uses $\mathcal R_{\textsc{ard}}$ in place of $\mathcal R$. Since the log-uniform prior is improper, Theorem~\ref{thm:ib_ard} does not apply to this model, and its ARD penalty serves as a regularizer. With a Gaussian position decoder $p_{\theta}(\mathbf{y}\mid\mathbf{z})$ of fixed isotropic variance, the negative log-likelihood in the second term of \eqref{eq:vib_obj} equals a squared-error regression loss up to additive and multiplicative constants, and the multiplicative constant is absorbed into the localization weight $\alpha$ in \eqref{eq:composite_loss}. In addition, we train a reconstruction decoder that maps the latent representation back to the feature space. Its output is used by the retrieval branch at the edge, and the reconstruction loss also regularizes the training.

\subsubsection{Orthogonality Under the IB Objective}\label{subsubsec:orthogonal_ib}

Under a tight communication constraint, only a few latent coordinates can be transmitted. If different coordinates encode the same visual direction, the transmitted coordinates are redundant and location-discriminative cues may be lost, which makes distinct places harder to separate. Let $\mathbf h_\phi(\mathbf x)\in\mathbb R^{d_h}$ be the final hidden feature of the encoder, where \(d_h=256\), and let $\mathbf W\in\mathbb R^{K\times d_h}$ be the weight matrix of the linear projection that maps \(\mathbf h_\phi(\mathbf x)\) to the posterior mean \(\boldsymbol\mu_\phi(\mathbf x)\). The rate penalty pulls posterior coordinates toward the prior, while the nested-prefix training in Section~\ref{subsubsec:scalable} favors early coordinates. Several rows of \(\mathbf W\) may then become aligned, which is called rank collapse and leaves fewer than \(K\) distinct latent directions for localization. We therefore impose approximate row-orthogonality on \(\mathbf W\). Proposition~\ref{proposition:vib_orthogonal} shows that if \(\mathbf{W}\mathbf{W}^\top\) is close to the identity, all singular values remain bounded away from zero and the mean projection avoids rank collapse across its \(K\) directions.

\begin{proposition}\label{proposition:vib_orthogonal}
Let $\mathbf{W}\in\mathbb{R}^{K\times d_h}$ denote the weight matrix of the posterior-mean projection. Assume the approximate orthogonality condition
\begin{equation}
\mathbf{W}\mathbf{W}^\top = \mathbf{I}_{K} + \boldsymbol{\Delta},
\end{equation}
where $\boldsymbol{\Delta}$ is a symmetric perturbation matrix satisfying $\|\boldsymbol{\Delta}\|_2 \le \varepsilon$ for some $0<\varepsilon<1$. Then all singular values $\sigma_i(\mathbf W)$ of $\mathbf{W}$ are bounded as
\begin{equation}
\sqrt{1-\varepsilon}\;\le\;\sigma_i(\mathbf{W})\;\le\;\sqrt{1+\varepsilon},
\qquad i=1,\ldots,K.
\end{equation}
In particular,
\begin{equation}
\sigma_{\min}(\mathbf{W}) \ge \sqrt{1-\varepsilon},
\end{equation}
so $\mathbf{W}$ has full row rank and the posterior-mean projection avoids rank collapse.
\end{proposition}

\begin{proof}
The singular values of $\mathbf{W}$ are the square roots of the eigenvalues of the \(K\times K\) matrix $\mathbf{W}\mathbf{W}^\top$. Since
\begin{equation}
\mathbf{W}\mathbf{W}^\top = \mathbf{I}_{K} + \boldsymbol{\Delta}
\quad \text{with} \quad \|\boldsymbol{\Delta}\|_2 \le \varepsilon,
\end{equation}
every eigenvalue of $\mathbf{W}\mathbf{W}^\top$ lies in the interval $[1-\varepsilon,\,1+\varepsilon]$. Therefore, for each singular value $\sigma_i(\mathbf{W})$,
\begin{equation}
1-\varepsilon \le \sigma_i^2(\mathbf{W}) \le 1+\varepsilon.
\end{equation}
Taking square roots gives the claimed bound.
\qed
\end{proof}

Since $\|\boldsymbol\Delta\|_2\le\|\boldsymbol\Delta\|_F$, the normalized Frobenius penalty \(K^{-2}\|\mathbf W\mathbf W^{\top}-\mathbf I_{K}\|_F^2\) in \eqref{eq:composite_loss} controls the perturbation \(\boldsymbol\Delta\), and Proposition~\ref{proposition:vib_orthogonal} applies when the trained \(\mathbf W\) satisfies \(\|\boldsymbol\Delta\|_2<1\). Its effect on the localization error is evaluated in Section~\ref{sec:scalable_eval}.

\subsubsection{Scalable View-Rate Encoding}\label{subsubsec:scalable}

The wireless link varies over time, and it is impractical to train and store a separate encoder and decoder for each link condition. We therefore train one encoder, in which a view mask selects the encoded camera views and a prefix length selects how many latent coordinates are transmitted. For a prefix length \(k\in\{1,\ldots,K\}\), the client transmits \(\mathbf z^{(k)}=[z_1,\ldots,z_k]\) and masks the remaining coordinates. A prefix contains \(4k\) bytes of float32 coordinates and a 20-byte representation header.

The prefix ordering is learned during training through nested rate dropout, following the ordered-representation principle of nested dropout~\cite{rippel2014nested} and Matryoshka representation learning~\cite{kusupati2022mrl}. For each minibatch, we sample a prefix length \(k\) from a finite training set \(\mathcal K\subseteq\{1,\ldots,K\}\), with \(K\in\mathcal K\), and apply the binary prefix mask
\begin{equation}
\mathbf m^{(k)}=[\underbrace{1,\ldots,1}_{k},\underbrace{0,\ldots,0}_{K-k}],
\qquad
\tilde{\mathbf z}^{(k)}=\mathbf m^{(k)}\odot \mathbf z .
\label{eq:prefix_mask}
\end{equation}
At the edge, the received prefix is padded with \(K-k\) zeros to form \(\tilde{\mathbf z}^{(k)}\). The same decoder is trained with these masked inputs. Since the coordinates are ordered, it can also decode a prefix whose length is outside \(\mathcal K\) at runtime. We call the design \emph{scalable} because one set of model parameters serves all transmitted lengths. If \(i<j\), coordinate \(i\) participates in at least as many prefix losses as coordinate \(j\), which encourages useful information to appear in the earlier coordinates.

We handle view adaptation with a view mask. Let \(\mathcal S_i\) be the set of admissible view modes of UAV client \(i\). For a UAV, \(\mathcal S_i=\mathcal S_{\mathrm{UAV}}=\{\mathrm F,\mathrm{FS},\mathrm{H4},\mathrm{A5}\}\), where F uses the front view, FS adds the left and right views, H4 uses the four horizontal views, and A5 uses all five views. For mode \(S\in\mathcal S_i\), inactive embeddings are set to zero before joint encoding, and the same view mask is applied to the decoded retrieval descriptor. Zeroing keeps one fixed encoder input for all view subsets. The application-layer request size is
\[
b_i(S,k)=4k+20+4|S|\ \text{bytes},
\]
where \(|S|\) is the number of selected views and each selected view adds a 4-byte identifier. Thus, one encoder supports the action set \(\mathcal A_i=\mathcal S_i\times\{1,\ldots,K\}\). The online controller uses the compact prefix \(k=16\) and the rich prefix \(k=32\), which form \(\mathcal K_{\mathrm{ctrl}}=\{16,32\}\), and evaluates the action set
\(\mathcal A_{i,\mathrm{ctrl}}=\mathcal S_i\times\mathcal K_{\mathrm{ctrl}}\).

\subsubsection{Overall Training Objective}

Let \(\mathbf m_S\in\{0,1\}^{M}\) be the view mask for mode \(S\), broadcast over the \(d\) coordinates of each view, and let \(\mathbf x_S\in\mathbb R^{d_f}\) be the resulting flattened feature. The encoder maps \(\mathbf x_S\) to \(\boldsymbol\mu_{\phi}(\mathbf x_S)\) and \(\boldsymbol\sigma_{\phi}(\mathbf x_S)\). Reparameterization produces \(\mathbf z\), and \eqref{eq:prefix_mask} produces \(\tilde{\mathbf z}^{(k)}\). The reconstruction decoder outputs \(\hat{\mathbf x}_{S,k}\in\mathbb R^{d_f}\), and the edge localizer outputs \(\hat{\mathbf y}_{S,k}\in\mathbb R^3\) from the reconstructed descriptor. By combining the reconstruction loss, the localization loss, the rate term, and the orthogonality penalty, with the view mode and the prefix length sampled during training, we obtain the overall training objective as
\begin{equation}
\begin{aligned}
\mathcal L(\phi,\theta)
=&\;
\mathbb E_{\mathbf x,S,k}
\!\left[
\left\|\mathbf x_S-\hat{\mathbf x}_{S,k}\right\|_2^2
\right]
+\alpha\,
\mathbb E_{\mathbf x,\mathbf y,S,k}
\!\left[
\left\|\mathbf y-\hat{\mathbf y}_{S,k}\right\|_2^2
\right]
\\[-1pt]
&+\beta\,
\mathbb E_{\mathbf x,S}
\!\left[
   \mathrm{KL}
\!\left(
q_\phi(\mathbf z\mid\mathbf x_S)
\,\Vert\,p(\mathbf z)
\right)
\right]
\\[-1pt]
&+\frac{\gamma}{K^2}
\left\|
\mathbf W\mathbf W^\top-\mathbf I_K
\right\|_F^2 ,
\end{aligned}
\label{eq:composite_loss}
\end{equation}
where \(S\in\mathcal S_{\mathrm{UAV}}\) is the sampled view mode and \(k\in\mathcal K\) is the sampled prefix length. The coefficients \(\alpha,\beta,\gamma>0\) balance the four terms. The second term in \eqref{eq:composite_loss} is the localization quality term of \eqref{eq:vib_obj} under the Gaussian position decoder, and the third term is its rate term. Thus, the second term improves the localization quality, while the third term keeps \(q_\phi(\mathbf z\mid\mathbf x_S)\) close to the prior \(p(\mathbf z)\). The reconstruction term preserves the retrieval cues, the prefix length determines the transmitted coordinates, and the orthogonality penalty constrains the rows of \(\mathbf W\).

\section{VOI-Guided Request and Mode Selection}
\label{sec:voi_control}

The scalable encoder in Section~\ref{sec:method} provides multiple view-rate modes, and the UAV client determines whether an edge correction is worth requesting in the current slot. The client does not observe the ground-truth localization error online, so it estimates the task-level VOI before transmission. The gross VOI is the predicted reduction in localization risk from an edge correction, and the net VOI reduces this predicted benefit by the communication and service costs. We denote the online controller by VOI-Control, and Fig.~\ref{fig:voi_timing} illustrates its request timing on the held-out route of Section~\ref{sec:eval}.

\begin{figure*}[t]
\centering
\includegraphics[width=0.98\textwidth]{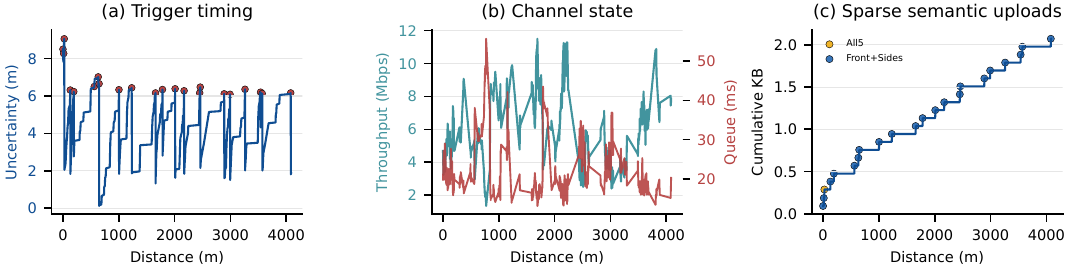}
\caption{VOI-Control request timing on a held-out route. The client delays edge requests while the local risk is low and triggers sparse uploads around uncertainty growth.}
\label{fig:voi_timing}
\vspace{-2mm}
\end{figure*}

\subsection{Local Belief and Predicted Edge Gain}

For client \(i\), let \(\bar{\mathbf y}_{i,t}\) and \(\mathbf P_{i,t}\) denote the locally propagated position estimate and its error covariance at slot \(t\), obtained from onboard odometry or a lightweight local estimator. The client navigates with \(\bar{\mathbf y}_{i,t}\) between edge corrections, while \(\mathbf P_{i,t}\) grows as odometry drift accumulates. We use
\begin{equation}
\mathcal L^{\mathrm{loc}}_{i,t}=\sqrt{\operatorname{tr}(\mathbf P_{i,t})}
\label{eq:local_risk}
\end{equation}
as a scalar local-risk proxy in meters. For each candidate action \(a=(S,k)\in\mathcal A_{i,\mathrm{ctrl}}\), the client predicts the residual risk after edge correction with
\begin{equation}
\widehat{\mathcal L}^{\mathrm{edge}}_{i,t}(a)
=g_{\psi}\!\left(c_{i,t},S,k\right),
\label{eq:edge_risk}
\end{equation}
where \(g_\psi\) is a residual predictor with parameters \(\psi\), and \(c_{i,t}\) is a locally available route-context category, such as corridor or intersection. In the evaluated replay, \(g_\psi\) returns the median edge residual observed on training routes for the same context and action \((S,k)\). The median can reduce the effect of rare high-residual samples and stabilize the action ranking.

\subsection{Net VOI and Online Action}

In decision theory, the value of a piece of information is the expected reduction in the decision loss obtained by acquiring it before acting~\cite{howard1966information}, and in networked control, a measurement is transmitted only when its value of information is nonnegative~\cite{Soleymani2022VoI}. Following this principle, we take the local estimate without edge assistance as the reference and define the predicted gross VOI of action \(a=(S,k)\) as
\begin{equation}
\Delta \widehat{\mathcal L}_{i,t}(a)
=\left[\mathcal L^{\mathrm{loc}}_{i,t}
-\widehat{\mathcal L}^{\mathrm{edge}}_{i,t}(a)\right]^+ .
\label{eq:gross_voi}
\end{equation}
Here, \([u]^+:=\max\{u,0\}\) denotes the positive part, which assigns zero gross VOI when the predicted residual is no smaller than the local-risk proxy.
For client \(i\) and action \(a=(S,k)\), let \(b(a)=b_i(S,k)\) be the application-layer payload in bytes, \(R_{i,t}\) be the current uplink rate in bits per second, \(E(a)\) be the estimated energy consumption of the client, and \(\widehat D_{i,t}(a)\) be the predicted queueing, service, and return delay excluding uplink transmission. By assigning nonnegative prices to the transmission time, energy, and service delay, we obtain the net VOI as
\begin{equation}
\widehat \nu_{i,t}(a)
=\Delta \widehat{\mathcal L}_{i,t}(a)
-\lambda_b\frac{8b(a)}{R_{i,t}}
-\lambda_e E(a)
-\lambda_d\widehat D_{i,t}(a),
\label{eq:net_voi}
\end{equation}
where \(\lambda_b,\lambda_e,\lambda_d\ge0\) convert the transmission time, energy, and delay into the same utility scale as the localization risk. The transmission-time term \(8b(a)/R_{i,t}\) is the uplink delay in \eqref{eq:problem}, so the communication constraint is enforced as a soft penalty. This score leads to a one-step decision rule that considers only the current slot. Our trace replay computes the transmission and delay terms from the payload and latency traces. The energy term applies only when an energy estimate of the client is available. The client chooses
\begin{equation}
a^*_{i,t}=\arg\max_{a\in\mathcal A_{i,\mathrm{ctrl}}}\widehat \nu_{i,t}(a),
\label{eq:action}
\end{equation}
and submits an edge request only if \(\widehat \nu_{i,t}(a^*_{i,t})>0\) and any configured local energy constraint is satisfied. Besides the compressed latent prefix, the request packet carries the selected view-rate mode, predicted VOI, payload size, and deadline metadata, which form the 20-byte header counted in \(b(a)\).

\section{VOI-Weighted Edge Scheduling}
\label{sec:voi_scheduler}

When multiple clients request edge localization simultaneously, the edge server allocates wireless and computing resources according to queue priority and localization value, as illustrated in Fig.~\ref{fig:edge_scheduler}. Let \(\mathcal I_t\) be the set of clients with at least one pending request at slot \(t\), and let \(Q_{i,t}\) be the pending-request count (queue length) for client \(i\). The head-of-line request of client \(i\) carries the net VOI \(\widehat\nu_{i,t}\) of its selected action in \eqref{eq:action}, which is predicted when the request is generated, the payload \(b_{i,t}\), and the estimated service time \(G_{i,t}\). The edge server chooses \(x_{i,t}\in\{0,1\}\), where 1 selects the head-of-line request for service, by solving
\begin{figure}[t]
\centering
\includegraphics[width=0.47\textwidth]{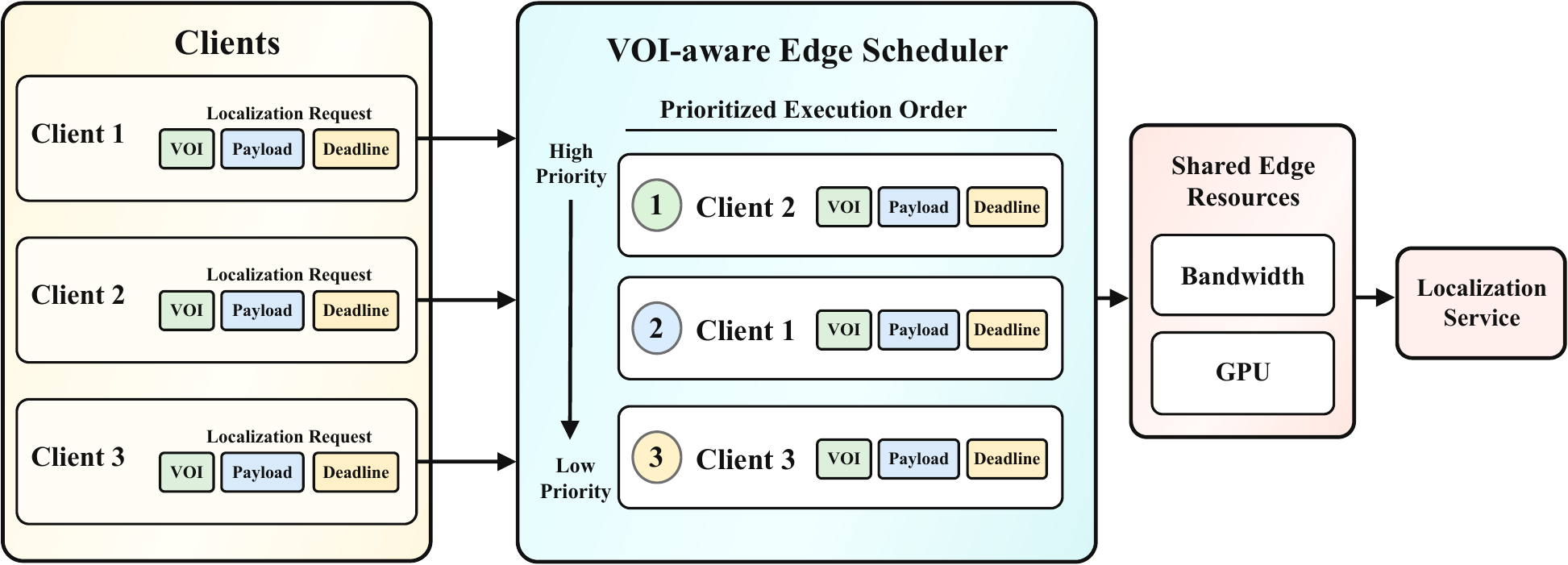}
\caption{VOI-weighted multi-client edge scheduling. Each request carries its localization value, payload, estimated service time, and deadline metadata.}
\label{fig:edge_scheduler}
\vspace{-2mm}
\end{figure}
\begin{equation}
\begin{aligned}
\max_{\{x_{i,t}\}}\quad
&\sum_{i\in\mathcal I_t}x_{i,t}\bigl(Q_{i,t}+V\widehat\nu_{i,t}\bigr)\\
\text{s.t.}\quad
&\sum_{i\in\mathcal I_t}x_{i,t}b_{i,t}\le b_t^{\max},\\
&\sum_{i\in\mathcal I_t}x_{i,t}G_{i,t}\le G_t^{\max},
\end{aligned}
\label{eq:voi_scheduler}
\end{equation}
where \(b_t^{\max}\) and \(G_t^{\max}\) are the per-slot budgets for wireless transmission and edge computation, respectively, and \(V>0\) controls the trade-off between the utility and the delay. The objective in \eqref{eq:voi_scheduler} is maximized because the DPP bound in the proof of Theorem~\ref{thm:voi_scheduler} contains its negative. Thus, minimizing the bound gives priority to the requests with a long queue and a high predicted VOI.

In the replay, one work-conserving edge server processes one request at a time. The baseline drift-plus-penalty scheduler as in~\cite{Li2024Timely}, denoted Base-DPP, selects the pending request with the largest priority \(Q_{i,t}+V\widehat\nu_{i,t}\), rather than solving the complete slotted allocation. The service model of the replay combines the mode-dependent compute profile with 3.1 ms for O-VIB decoding, 5.2 ms for retrieval, 4.0 ms of runtime overhead, and lognormal jitter with a 2.2-ms median. These values are the statistics measured in our experiments, and the 3.1-ms decoding time is also reported in Fig.~\ref{fig:single_scene_ib_ablation}(d). The reported edge latency includes the queueing time before service.
Each queue evolves as
\begin{equation}
Q_{i,t+1}=\left[Q_{i,t}-x_{i,t}\right]^+ + A_{i,t},
\label{eq:queue}
\end{equation}
where \(A_{i,t}\) indicates whether a new request from client \(i\) is admitted into the edge queue. The rule in \eqref{eq:voi_scheduler} raises priority with backlog and predicted localization value. However, it does not explicitly consider the waiting age or impending deadline misses.

We therefore design a shaped scheduler as a heuristic extension of \eqref{eq:voi_scheduler}. It supplements the predicted VOI with the current risk, the drift, and the mission urgency, so that a UAV with a fast drift or a high local risk receives a higher priority. For queued request \(j\) with net VOI \(\widehat\nu_{j,t}\), estimated service time \(G_{j,t}\), and local risk \(\mathcal L_j^{\mathrm{loc}}\), we define the mission-weighted scheduling value as
\begin{equation}
\begin{aligned}
\widehat v_{j,t}=\max\Bigl\{0.1,\omega_j\bigl(&\widehat\nu_{j,t}+1.8r_j^{\mathrm{drift}}\\
&+0.22\mathcal L_j^{\mathrm{loc}}-0.006\widehat D_{j,t}\bigr)\Bigr\},
\end{aligned}
\label{eq:scheduler_value}
\end{equation}
where \(\omega_j\) is the mission-urgency weight of request \(j\), \(r_j^{\mathrm{drift}}\) is the fitted drift rate in meters per second, and \(\widehat D_{j,t}\) is the predicted latency in milliseconds. The coefficients 1.8 s, 0.22, and 0.006 m/ms convert the drift, risk, and latency terms to the meter scale of \(\widehat\nu_{j,t}\), and the floor of 0.1 m keeps the scheduling value positive. These coefficients are hyperparameters, which were determined based on simulation tests and fixed before testing. Let \(\operatorname{nrm}_t(u_j)\) denote min-max normalization over the queued requests, with a zero output when the range is zero. The priority score is
\begin{equation}
\begin{aligned}
s_{j,t}={}&
w_v\operatorname{nrm}_t\!\left(\frac{\widehat v_{j,t}}{G_{j,t}}\right)
+w_r\operatorname{nrm}_t(\widehat v_{j,t})
+w_\ell\bigl(1-\operatorname{nrm}_t(\ell_{j,t})\bigr)\\
&+w_a\operatorname{nrm}_t(\zeta_{j,t})
+w_o\operatorname{nrm}_t(o_{j,t})
+w_s\operatorname{nrm}_t(-G_{j,t}).
\end{aligned}
\label{eq:voi_scheduler_impl}
\end{equation}
Let \(t^{\mathrm{now}}\) be the clock time at scheduling decision \(t\), \(t_j^{\mathrm{rdy}}\) the ready time, and \(d_j\) the absolute deadline, all expressed in the same time unit as the relative deadline \(D_j^{\mathrm{ddl}}\). The replay uses \(\ell_{j,t}=[d_j-t^{\mathrm{now}}-G_{j,t}]^+\), \(\zeta_{j,t}=[t^{\mathrm{now}}-t_j^{\mathrm{rdy}}]^+/D_j^{\mathrm{ddl}}\), and \(o_{j,t}=[t^{\mathrm{now}}+G_{j,t}-d_j]^+/D_j^{\mathrm{ddl}}\). The laxity \(\ell_{j,t}\) measures the remaining service margin, \(\zeta_{j,t}\) is the normalized waiting age, and \(o_{j,t}\) measures the predicted completion overrun. The hyperparameters are set to \((w_v,w_r,w_\ell,w_a,w_o,w_s)=(3.2,1.2,1.2,1.4,0.6,0.1)\), which were determined based on simulation tests. The largest weight is given to the value per unit service time, so that the edge server prefers the requests that bring a high value with a short service time. We evaluate the single-request Base-DPP rule described above and the shaped scheduler, denoted VOI-Lyapunov.

\begin{theorem}\label{thm:voi_scheduler}
Suppose that the arrival indicators \(A_{i,t}\), payloads \(b_{i,t}\), service times \(G_{i,t}\), predicted values \(\widehat\nu_{i,t}\), and budgets \(b_t^{\max}\) and \(G_t^{\max}\) are uniformly bounded, independent and identically distributed across slots, and independent of the queue state. Suppose that \eqref{eq:voi_scheduler} is solved exactly in each slot and that a stationary randomized policy satisfies the per-slot budgets while serving every client at a rate above its mean arrival rate by a slack \(\delta>0\). Then every queue under \eqref{eq:voi_scheduler} is mean-rate stable. The time-average predicted-VOI utility
\(\liminf_{T\to\infty}T^{-1}\sum_{t=0}^{T-1}\sum_i\mathbb E[x_{i,t}\widehat\nu_{i,t}]\)
is within \(O(1/V)\) of the largest utility achievable by any policy that satisfies the per-slot budgets and stabilizes all queues, and the \(\limsup\) of the time-average total backlog is \(O(V)\).
\end{theorem}

\begin{proof}
Let \(L(\mathbf Q_t)=\frac{1}{2}\sum_i Q_{i,t}^2\) be the quadratic Lyapunov function and let \(\Delta(\mathbf Q_t)\) denote its one-slot conditional drift. Using the queue update in \eqref{eq:queue} and the bounded arrivals and services gives
\begin{equation}
\Delta(\mathbf Q_t)\le C_0+\sum_i Q_{i,t}\,\mathbb E[A_{i,t}-x_{i,t}\mid\mathbf Q_t]
\label{eq:drift}
\end{equation}
for a finite constant \(C_0\) determined by the bounds on arrivals and services. Subtracting \(V\,\mathbb E[\sum_i x_{i,t}\widehat\nu_{i,t}\mid\mathbf Q_t]\) from both sides yields the drift-plus-penalty expression. The scheduler in \eqref{eq:voi_scheduler} minimizes its right-hand side in every slot because the arrival terms do not depend on the decision. Comparing this minimum with the stationary randomized policy, taking expectations, and summing over time give mean-rate stability, an \(O(1/V)\) utility gap, and an \(O(V)\) backlog bound by the standard drift-plus-penalty argument~\cite{neely2010sno}.
\qed
\end{proof}

Theorem~\ref{thm:voi_scheduler} holds for the exact slotted optimizer under the stated independence assumptions. It does not cover the single-request replay or the shaped scheduler, where the request attributes persist until service and the arrivals depend on the queue state through \eqref{eq:net_voi}. Section~\ref{sec:sched_eval} therefore evaluates the shaped scheduler empirically.

\section{Performance Evaluation}\label{sec:eval}

\subsection{Experimental Setup}

\begin{figure}[t]
\centering
\includegraphics[width=\columnwidth]{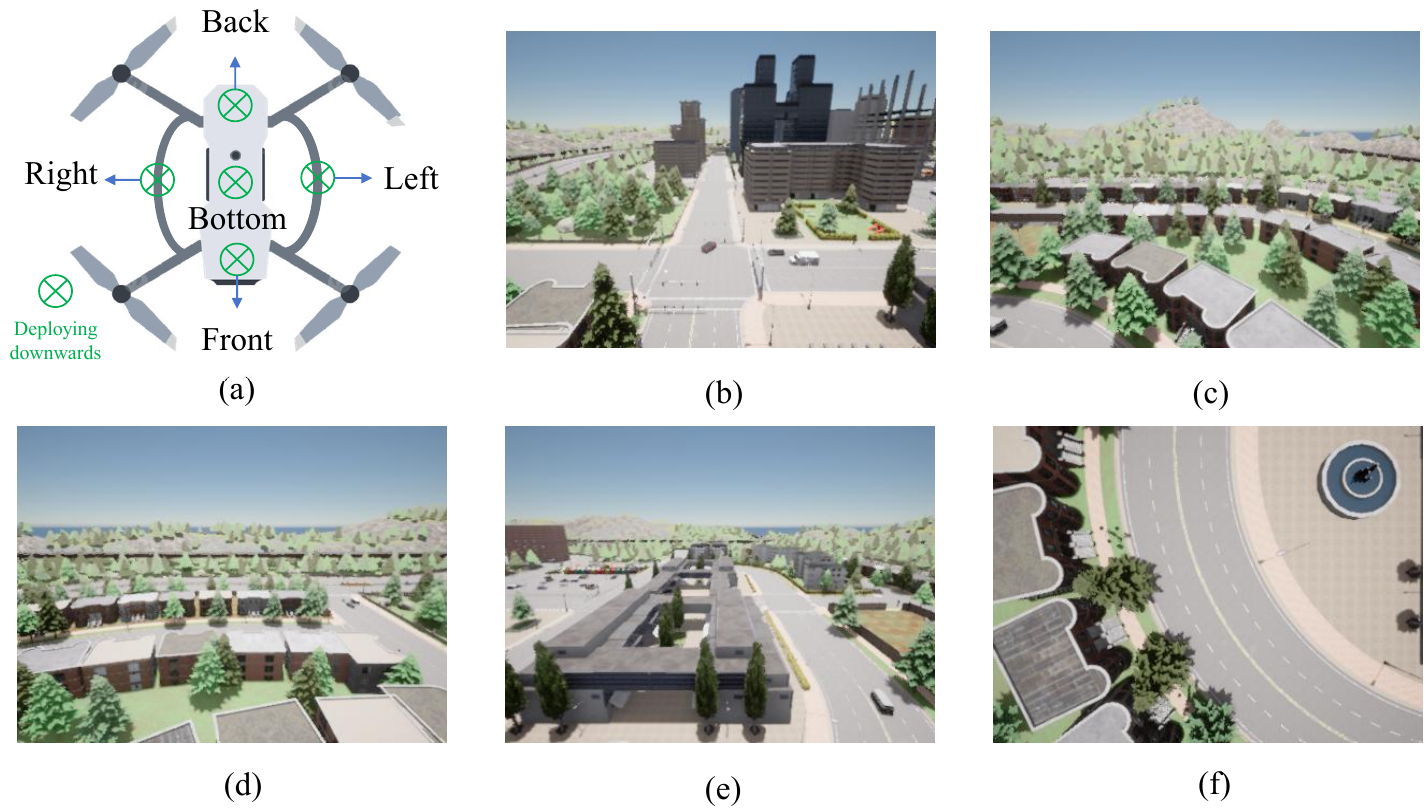}
\caption{Self-collected multi-scene multi-view UAV data rendered in the CARLA simulator. The synchronized five-view RGB observations and logged simulator poses support retrieval, feature extraction, and model training.}
\label{fig:testbed_dataset}
\vspace{-2mm}
\end{figure}

\begin{table*}[t]
\centering
\caption{Datasets and evaluation configurations used in the study.}
\label{tab:setup}
\scriptsize
\setlength{\tabcolsep}{4pt}
\begin{tabularx}{\textwidth}{@{}lllcclX@{}}
\toprule
Setting & Data path & Scenes & Frames/routes & Views & Ground truth & Main purpose \\
\midrule
Single-scene split & Conference-version Town05 & 1 scene & 4,994; 3,495/749/750 & 5 & Simulator pose & Scalable compression and ablations \\
Raw multi-scene data & CARLA UAV & 5 scenes & 6,684 frames & 5 & Simulator pose & Raw samples and retrieval database \\
Sampled benchmark & CARLA subset & 5 scenes & 4,593 frames & 5 & Simulator pose & Coarse-to-fine retrieval and plots \\
Control workload & Town02 route replay & Held-out route & 593 slots / 7 policies & Variable & Simulator pose & \(g_\psi\) request and mode control \\
Retrieval queries & Offline query set & 5 scenes & 400 queries & 5 & Simulator pose & Scene coarse + tile-pruned fine search \\
Indoor UAV & Qualisys capture & 4.1 m $\times$ 4.0 m & 76 ref. / 58 query & 5 sequential & 6-DoF mocap & Real-scene compact localization \\
Indoor UGV & Qualisys capture & 5.0 m $\times$ 4.0 m & 78 ref. / 72 query & 4 sequential & 6-DoF mocap & Real-scene compact localization \\
\bottomrule
\end{tabularx}

\vspace{-2mm}
\end{table*}

Table~\ref{tab:setup} and Fig.~\ref{fig:testbed_dataset} summarize the data and evaluation settings. Per-mode residuals and payloads come from trained-model inference. Semantic-path latency is measured on a Jetson Orin NX after feature extraction. The request-control experiments additionally simulate odometry drift, request timing, and network delay. Unless otherwise specified, the localization error is the Euclidean distance to the ground truth, and the application-layer request adds selected-view identifiers to the serialized \(4k+20\)-byte semantic representation.\footnote{The multi-view UAV dataset was collected by the authors and is released at \url{https://huggingface.co/datasets/Peter341/Multi-View-UAV-Dataset}. The code is available at \url{https://github.com/fangzr/TOC-Edge-Aerial}.}

Compression experiments fix the downstream localizer after training, which we call the frozen-localizer protocol. Controller configurations are selected on tuning seeds and evaluated on disjoint held-out seeds, where each seed fixes one realization of the channel, drift, request timing, and queue state.

The Town05 scalable encoder is trained with \(K=128\), the seven prefix lengths \(\mathcal K=\{4,8,16,32,64,96,128\}\), 50 AdamW epochs, batch size 512, learning rate \(10^{-3}\), and \((\alpha,\beta,\gamma)=(1,10^{-3},0.1)\). Raw-image retrieval uses a 2 Mbps uplink. Semantic-payload experiments use 4 to 12 KB/s profiles. Request replay uses 0.5-s slots and a shared edge queue.

\subsection{Scalable Compression Performance}\label{sec:scalable_eval}

\begin{figure*}[t]
\centering
\includegraphics[width=0.98\textwidth]{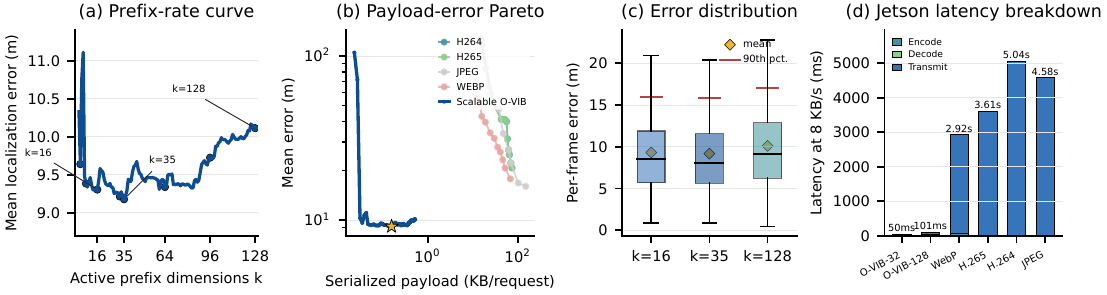}
\caption{Scalable O-VIB compression on Town05. Panel (a) evaluates one trained model over the prefix length \(k\). Panel (b) compares serialized payload-error frontiers with image-codec experiments. Panel (c) reports per-frame errors. Diamonds mark means and red ticks mark the 90th-percentile (p90) values. Panel (d) reports post-feature-extraction Jetson latency at 8 KB/s.}
\label{fig:scalable_compression}
\vspace{-2mm}
\end{figure*}

\begin{table}[t]
\centering
\caption{Deployment overhead of fixed-rate and scalable O-VIB configurations under the Jetson Orin NX latency profile.}
\label{tab:scalable_overhead}
\scriptsize
\setlength{\tabcolsep}{3.5pt}
\resizebox{\linewidth}{!}{%
\begin{tabular}{lcccc}
\toprule
Method & Models & Storage & Rates & Latency @ 8 KB/s \\
\midrule
Fixed-rate O-VIB-32 & 1 & 2.8 MB & 1 & 50.4 ms \\
Fixed-rate O-VIB-128 & 1 & 2.9 MB & 1 & 100.6 ms \\
Fixed-rate bank $\{32,128\}$ & 2 & 5.7 MB & 2 & 50.4--100.6 ms \\
Scalable O-VIB $k=35$ & 1 & 13.9 MB & 128 & 52.0 ms \\
Scalable O-VIB $k=16$ & 1 & 13.9 MB & 128 & 42.0 ms \\
\bottomrule
\end{tabular}}
\par\vspace{3pt}
\begin{minipage}{0.96\linewidth}
\footnotesize \textbf{Note:} Latency combines the measured Jetson compute profile with transmission of the serialized $4k+20$-byte O-VIB payload at 8 KB/s. The Rates column counts prefix lengths supported at runtime, not independently trained models.
\end{minipage}

\vspace{-2mm}
\end{table}

Fig.~\ref{fig:scalable_compression} uses the Town05 split in Table~\ref{tab:setup}. Table~\ref{tab:scalable_overhead} reports the deployment overhead. The model is trained with the seven prefix lengths in \(\mathcal K\) and evaluated at every \(k=1,\ldots,128\), so each point in panel (a) is a measured result. Validation selects \(k=35\), which lies between the training lengths 32 and 64 and gives 9.18 m test error at 0.156 KB/request. Prefix \(k=16\) gives 9.31 m at 0.082 KB and \(k=128\) gives 10.11 m at 0.520 KB. At 8 KB/s after feature extraction, the fixed-rate O-VIB-32 and O-VIB-128 models in Table~\ref{tab:scalable_overhead} take 50.4 and 100.6 ms, respectively, compared with 2.92 s for WebP, 3.61 s for H.265, 5.04 s for H.264, and 4.58 s for JPEG.

\begin{figure*}[t]
\centering
\includegraphics[width=0.98\textwidth]{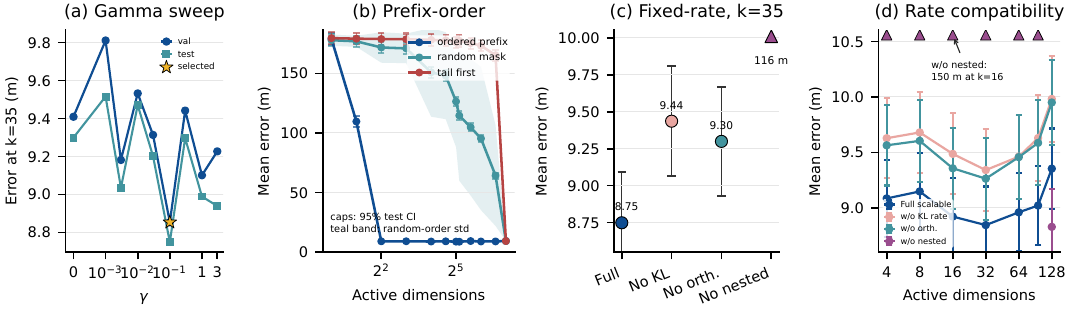}
\caption{Scalable encoder ablations on Town05 with the frozen-localizer protocol. Panel (a) sweeps \(\gamma\) at \(k=35\), with the star marking the validation-selected positive-\(\gamma\) Full model. Panel (b) reports 95\% confidence intervals over 750 test frames. The teal band additionally shows the standard deviation across eight random coordinate permutations. Panel (c) compares components at \(k=35\). Panel (d) zooms into 8.6 to 10.6 m.}
\label{fig:scalable_encoder_ablation}
\vspace{-2mm}
\end{figure*}

Fig.~\ref{fig:scalable_encoder_ablation} reports separately trained ordering and component ablations under the same frozen-localizer protocol. Validation selects \(\gamma=0.1\) at \(k=35\), giving 8.85 m validation and 8.75 m test error. This checkpoint is independent of the 9.18 m model in Fig.~\ref{fig:scalable_compression}. At \(k=35\), the full model, denoted Full, obtains 8.75 m, compared with 9.44 m without the rate term, 9.30 m without orthogonality, and 116.19 m without nested-prefix supervision. Random and tail-first coordinate orderings remain worse than the learned prefix ordering until most coordinates are retained. Larger \(\gamma\) values do not improve the validation error, and the model without nested-prefix training collapses at short prefixes.

\begin{figure*}[t]
\centering
\includegraphics[width=0.98\textwidth]{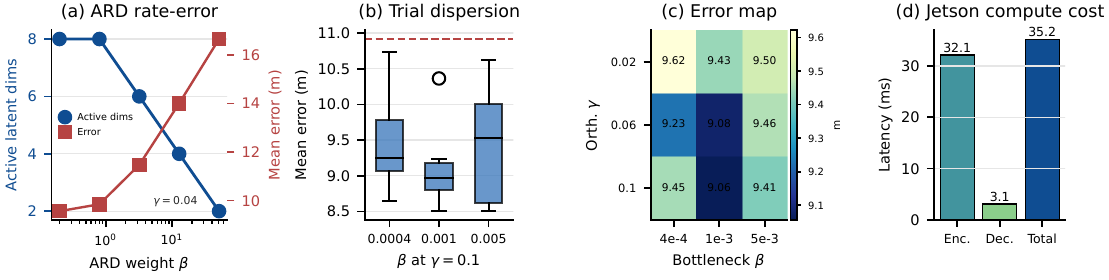}
\caption{Reproduction of the conference-version single-scene O-VIB experiments. Panel (a) shows ARD-controlled active dimensions and error. Panels (b) and (c) report the fixed-\(k=16\) sweep, and panel (d) reports Jetson-profile compute cost.}
\label{fig:single_scene_ib_ablation}
\vspace{-2mm}
\end{figure*}

Fig.~\ref{fig:single_scene_ib_ablation} summarizes the conference-version single-scene module results. It uses the same Town05 split as Fig.~\ref{fig:scalable_compression}. Panel (a) shows that increasing the ARD weight reduces the active dimensions from 8 to 2 and raises the mean error from 9.57 m to 16.65 m. In the fixed-\(k=16\) sweep, the best reproduced O-VIB setting is \(\gamma=0.1\) and \(\beta=0.001\), with 9.06 m mean error over 10 trials, compared with 10.92 m for the conference-version multi-view fusion baseline. This fixed-rate model uses the log-uniform prior, and the experiment isolates the effects of ARD and orthogonality from scalable-prefix training.

\subsection{Coarse-to-Fine Retrieval}\label{sec:retrieval_eval}

\begin{figure*}[t]
\centering
\includegraphics[width=0.98\textwidth]{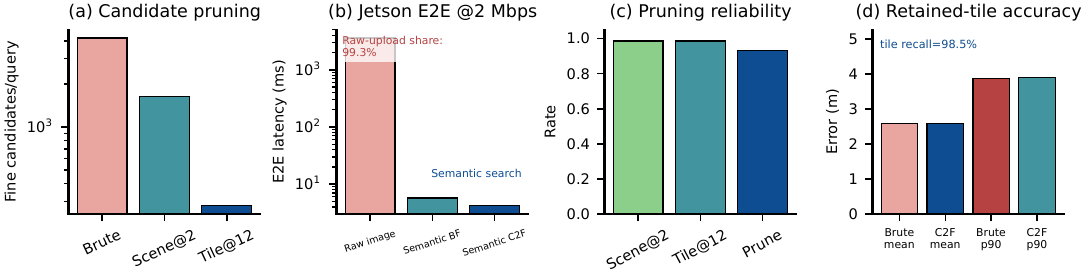}
\caption{Coarse-to-fine retrieval on the sampled multi-scene benchmark. Scene and tile prototypes reduce the fine-search set. Panel (c) reports the pruned fraction, and panel (d) reports accuracy conditional on retaining the correct tile. All-query accuracy is given in the text.}
\label{fig:retrieval_pruning}
\vspace{-2mm}
\end{figure*}

Fig.~\ref{fig:retrieval_pruning} evaluates 400 sampled multi-scene queries. The mean raw five-view payload is 891.5 KB. On a Jetson Orin NX 16 GB with a 2 Mbps uplink, raw-image offloading takes 3.59 s, with 99.3\% spent on upload. Semantic brute-force and coarse-to-fine retrieval take 5.75 and 4.21 ms after upload, respectively, with 93.3\% of fine-search descriptors pruned by the latter. The coarse stage keeps the two most similar scene prototypes and the twelve most similar tile prototypes, and the scene top-2 accuracy and tile top-12 recall are both 98.5\%. Conditional on retaining the correct tile, both methods obtain 2.60 m mean error. Over all queries, coarse-to-fine rises to 6.52 m because six coarse-stage misses select the wrong scene, whereas brute force remains at 2.60 m.

\subsection{View-Rate Selection}\label{sec:voi_evaluation}

\begin{figure*}[t]
\centering
\includegraphics[width=0.98\textwidth]{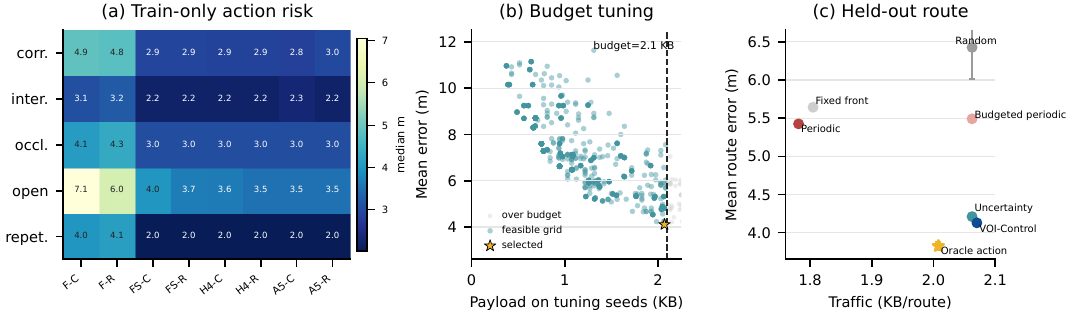}
\caption{Calibration and evaluation of VOI-Control. Panel (a) shows the train-only context/action median residual. C and R denote \(k=16\) and \(k=32\). Panel (b) shows tuning-seed budget search, and panel (c) shows held-out policies with the oracle action.}
\label{fig:voi_predictor}
\vspace{-2mm}
\end{figure*}

\begin{figure*}[t]
\centering
\includegraphics[width=0.98\textwidth]{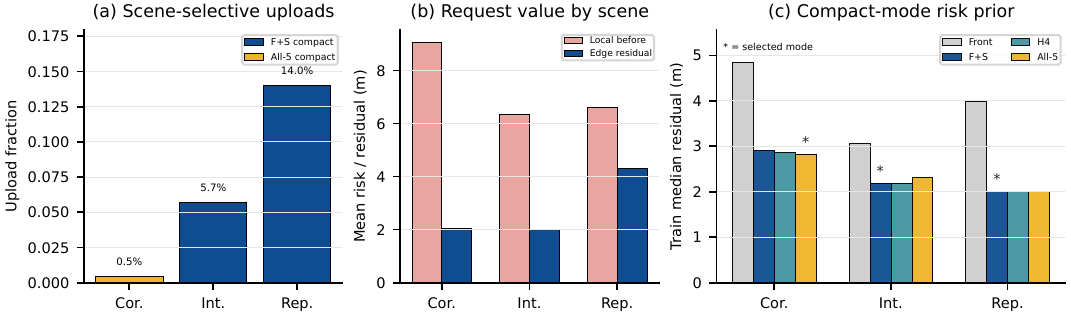}
\caption{View-mode behavior under the budget-constrained held-out route. Panel (a) reports upload fractions by compact view mode, panel (b) compares pre-request risk with returned residuals, and panel (c) shows the train-only compact-mode prior.}
\label{fig:view_rate_selection}
\vspace{-2mm}
\end{figure*}

We compare VOI-Control with five causal request policies under the 2.1 KB/route budget and include an oracle action bound. Periodic sends an FS compact request every 16 s. Budgeted periodic places the same number of FS compact requests as VOI-Control at uniformly spaced slots. Fixed-front sends F compact requests when the local risk exceeds 8 m, with a minimum gap of 1 s. Random uses the same request count and FS compact mode at uniformly sampled slots. Uncertainty sends FS compact requests when the local risk exceeds 6 m, with a minimum gap of 6 s. Oracle action uses the request times selected by VOI-Control and chooses the action with the lowest realized residual after applying the same latency and payload penalties.

Fig.~\ref{fig:voi_predictor} analyzes VOI-Control on Town02 routes containing corridor, intersection, and repetitive scenes. The prefix checkpoint is trained on frames 000000 to 000355 with masked reconstruction, four UAV view modes, and the training prefix lengths \(\{16,32,35,64,128\}\). It uses neither position labels nor held-out route frames. Candidate residuals use leave-one-out retrieval with the same-scene database, and the median prior is fitted on training routes. Tuning seeds select the lowest-risk feasible configuration, which uses compact prefixes for admitted requests. Across five held-out seeds with 95\% confidence intervals, VOI-Control obtains \(4.13\pm0.03\) m mean and \(6.71\pm0.08\) m p95 route error, compared with \(4.21\pm0.03\) m and \(7.01\pm0.10\) m for uncertainty triggering. Budgeted periodic offloading obtains \(5.49\pm0.03\) m mean and \(9.73\pm0.09\) m p95 route error, so VOI-Control reduces these two errors by 24.8\% and 31.0\%, respectively. Relative to uncertainty triggering, the reductions are 1.9\% and 4.3\%, respectively. Budgeted periodic matches the request count and compact prefix but fixes the view mode to FS, so this comparison evaluates request timing and view selection jointly. The oracle action reaches 3.82 m mean route error.

Fig.~\ref{fig:view_rate_selection} shows how the selected VOI-Control configuration spends its traffic on route segments with nonzero uploads. On the held-out route, VOI-Control sends rare all-five compact requests in corridor scenes and front-plus-sides compact requests in intersection and repetitive scenes. It uploads on 0.5\% of corridor steps, 5.7\% of intersection steps, and 14.0\% of repetitive steps. These requests reduce the local risk from roughly 6 to 9 m down to 2 to 4 m, depending on the scene. The compact-mode prior in panel (c) explains the mode choices. Front-only compact has much larger residuals, while front-plus-sides and all-five compact are close to the best compact residual in scenes where uploads are selected.

\subsection{Request-Control Dynamics}\label{sec:request_eval}

\begin{figure*}[t]
\centering
\includegraphics[width=0.98\textwidth]{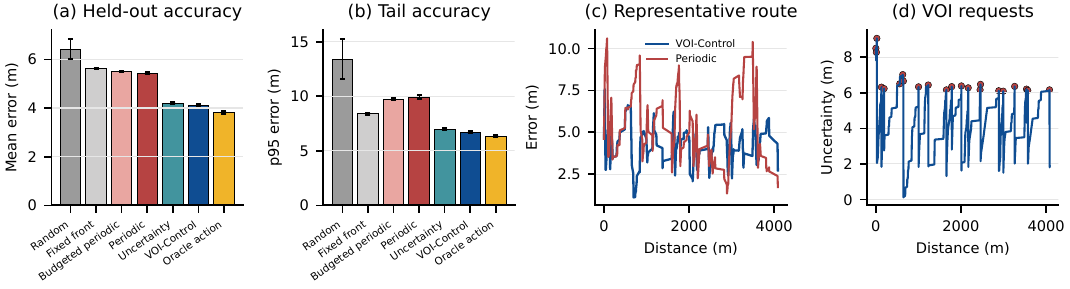}
\caption{Budget-constrained request-control dynamics on the held-out route. Panels report held-out mean error, held-out p95 error, a representative error curve, and the corresponding VOI-Control request timing.}
\label{fig:request_control}
\vspace{-2mm}
\end{figure*}

Fig.~\ref{fig:request_control} shows the same budget-constrained behavior over time. VOI-Control issues 22 requests and transmits 2.07 KB per route. Under the same budget, fixed-front obtains 5.64 m mean and 8.41 m p95 error, while random obtains 6.43 m and 13.43 m. The other policies and confidence intervals appear in Fig.~\ref{fig:voi_predictor}(c).

\subsection{Multi-Client Scheduling}\label{sec:sched_eval}

\begin{figure*}[t]
\centering
\includegraphics[width=0.98\textwidth]{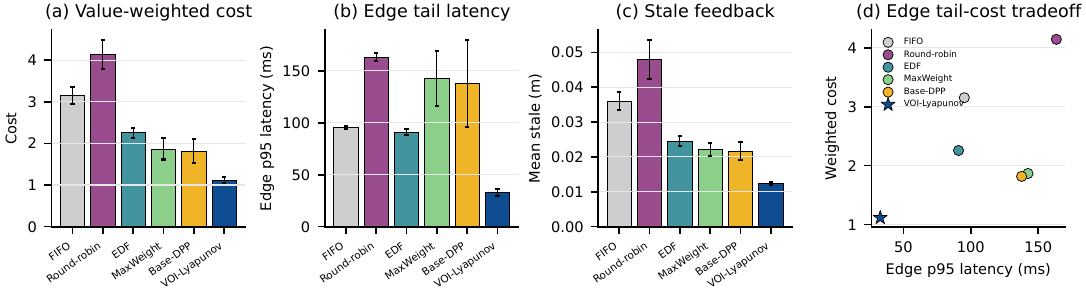}
\caption{Value-aware scheduling under high congestion with \(N=30\) UAV clients and a 4 KB/s weak link. Panels report cost and latency metrics for the top-10\% high-value requests.}
\label{fig:multi_client_scheduling}
\vspace{-2mm}
\end{figure*}

\begin{table}[t]
\centering
\caption{High-congestion scheduler comparison for top-10\% high-value requests.}
\label{tab:multi_client}
\scriptsize
\setlength{\tabcolsep}{1.5pt}
\begin{tabular*}{\columnwidth}{@{\extracolsep{\fill}}lrrrr@{}}
\toprule
\multicolumn{5}{c}{\textbf{\makecell{High congestion: $N=30$\\4 KB/s weak link}}} \\
\midrule
Scheduler & E2E p95 & Edge p95 & Stale & Cost \\
\midrule
MaxWeight & 198 ms & 142.4 ms & 0.022 m & 1.87 \\
Base-DPP & 193 ms & 137.7 ms & 0.022 m & 1.82 \\
\textbf{VOI-Lyap.} & \textbf{88.5 ms} & \textbf{32.8 ms} & \textbf{0.012 m} & \textbf{1.12} \\
\bottomrule
\end{tabular*}
\par\vspace{3pt}
\begin{minipage}{0.98\linewidth}
\footnotesize \textbf{Note:} All table metrics, including stale-feedback penalty, are computed over the top-10\% high-value requests.
\end{minipage}

\vspace{-2mm}
\end{table}

We obtain the multi-client results in Fig.~\ref{fig:multi_client_scheduling} and Table~\ref{tab:multi_client} by replaying VOI-Control request events from the held-out route on a single shared edge server. The replay covers low, medium, and high congestion, corresponding to \(N=10\), \(N=20\), and \(N=30\) UAV clients with weak-link profiles of 12, 8, and 4 KB/s, respectively. Each client count uses five scheduler seeds that perturb route phase, mission urgency, deadlines, and service-time jitter while preserving trained-model payloads and localization residuals. The weighted cost is \(\widehat v_jD_j/D_j^{\mathrm{ddl}}\), where \(D_j\) is the response latency and \(D_j^{\mathrm{ddl}}\) is the request deadline. The stale-feedback penalty is \(r_j^{\mathrm{drift}}D_j^{\mathrm{stale}}/1000\) meters, where \(D_j^{\mathrm{stale}}\) is the observation-to-correction delay in milliseconds. The top-10\% weighted cost is computed over the largest decile of \(\widehat v_j\).

Each application-layer request is 0.094 to 0.102 KB. The p95 transmission times under the three profiles are 8.3, 12.7, and 25.4 ms, respectively. For each request, the replay adds the Jetson O-VIB-32 encoding profile and weak-link transmission time to edge queueing and service latency before computing the end-to-end p95.

Base-DPP uses \(V=8\) in \eqref{eq:voi_scheduler}, while VOI-Lyapunov applies \eqref{eq:voi_scheduler_impl}. We also compare first-in-first-out (FIFO), earliest-deadline-first (EDF), and max-weight scheduling~\cite{Liu2023TrackingMaxWeight}, denoted MaxWeight, which serves the longest queue and corresponds to \(V=0\) in \eqref{eq:voi_scheduler}. For the top-10\% high-value requests, VOI-Lyapunov reduces the weighted cost to 1.12, compared with 1.82 for Base-DPP and 1.87 for MaxWeight. It lowers the edge-side p95 latency from 137.7 ms for Base-DPP to 32.8 ms and the end-to-end p95 latency to 88.5 ms. The corresponding end-to-end values of Base-DPP and MaxWeight are 193.0 ms and 197.8 ms, respectively. Across all high-congestion requests, VOI-Lyapunov gives 185.7 ms edge and 241.1 ms end-to-end p95, compared with 232.5/288.0 ms for Base-DPP and 217.8/273.2 ms for MaxWeight. FIFO and EDF give 93.2/148.8 ms and 106.7/162.5 ms for edge/end-to-end p95, respectively, because they prioritize latency rather than request value.

\subsection{Real-Scene UAV Localization and UGV Extension}\label{sec:real_scene_evaluation}

\begin{figure*}[t]
\centering
\includegraphics[width=0.98\textwidth]{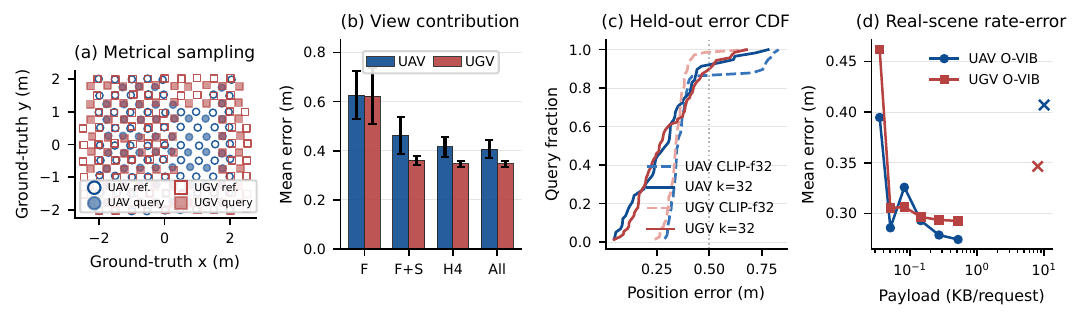}
\caption{Indoor UAV localization and supplementary UGV results with motion-capture ground truth. (a) Reference and disjoint query locations: circles denote UAV, squares denote UGV, and open/filled markers denote reference/query samples. (b) Mean validation-calibrated frozen-CLIP retrieval error for front only (F), front plus sides (F+S), all four horizontal directions (H4), and all available views (All). Caps are 95\% confidence intervals over queries. (c) Held-out query-error CDFs for all-view CLIP retrieval and O-VIB at \(k=32\). (d) Mean query error versus the actual serialized representation size. Crosses denote the uncompressed float32 CLIP descriptors.}
\label{fig:real_system}
\vspace{-2mm}
\end{figure*}

\begin{figure}[t]
\centering
\includegraphics[width=0.75\columnwidth]{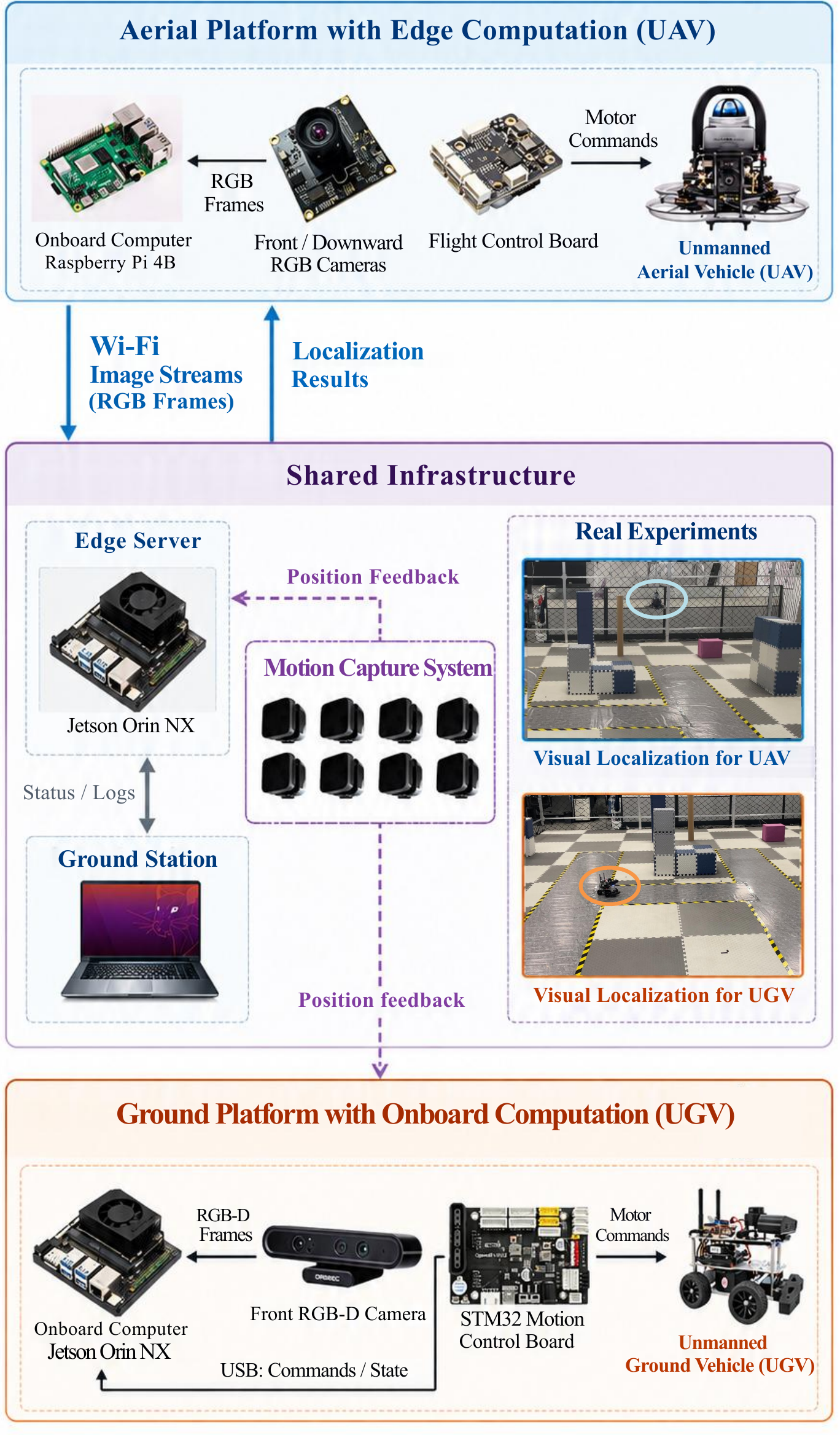}
\caption{Indoor UAV localization testbed, with a ground platform used for supplementary validation of the same design. Visual inputs are sent over Wi-Fi, and Qualisys supplies pose ground truth. The UAV and the UGV are evaluated separately.}
\label{fig:real_testbed_photos}
\vspace{-2mm}
\end{figure}

We first validate the localization pipeline on the indoor UAV testbed in Fig.~\ref{fig:real_testbed_photos}, and then use a UGV as a supplementary ground-platform extension. The UAV dataset contains 76 reference and 58 query locations with five pose-aligned views, and the UGV dataset contains 78 reference and 72 query locations with four horizontal views. The UGV uses $\mathcal S_{\mathrm{UGV}}=\{\mathrm F,\mathrm{FS},\mathrm{H4}\}$, where H4 includes all its views, and its model is trained with the UGV reference data. Horizontal directions were captured sequentially near each target pose rather than with a synchronized camera array, while the UAV front and downward views were captured together. Reference samples form the localization database and are split spatially for encoder training and validation, while disjoint query samples are used once for reporting. The mixing weight \(\eta\) is selected on the reference validation subset for each prefix. The CLIP baseline uses pure nearest-neighbor retrieval with nonnegative view-fusion weights calibrated on five spatial folds of the reference split and renormalized over the retained views.

\emph{UAV localization.} Fig.~\ref{fig:real_system}(b) first isolates the contribution of additional views without learned compression. Relative to front-only retrieval, calibrated H4 reduces the UAV mean error from 0.626 to 0.415 m, and all five views further reduce it to 0.407 m. The mean distance from each UAV query to its nearest reference location is 0.322 m. This spacing lower-bounds the mean error of pure retrieval, which always returns a reference location, whereas regression can interpolate between reference locations.

\begin{table}[!t]
\centering
\caption{Held-out UAV localization and supplementary UGV results.}
\label{tab:real_ablation}
\footnotesize
\setlength{\tabcolsep}{1.2pt}
\begin{tabular*}{\columnwidth}{@{\extracolsep{\fill}}llrrrr@{}}
\toprule
Plat. & Representation & KB & Mean & p90 & R@0.5 \\ 
 & & & (m) & (m) & \\ 
\midrule
UAV & CLIP & 10.000 & 0.407 & 0.758 & 86.2\% \\
UAV & O-VIB-8 & 0.051 & 0.286 & 0.458 & 89.7\% \\
UAV & O-VIB-32 & 0.145 & 0.293 & 0.447 & 91.4\% \\
UAV & O-VIB-128 & 0.520 & 0.274 & 0.459 & 93.1\% \\
UGV & CLIP & 8.000 & 0.346 & 0.396 & 98.6\% \\
UGV & O-VIB-8 & 0.051 & 0.305 & 0.519 & 88.9\% \\
UGV & O-VIB-32 & 0.145 & 0.296 & 0.501 & 88.9\% \\
UGV & O-VIB-128 & 0.520 & 0.293 & 0.481 & 90.3\% \\
\bottomrule
\end{tabular*}
\par\vspace{3pt}
\begin{minipage}{\columnwidth}
\scriptsize \textbf{Note:} Reference samples train O-VIB, form the localization database, and calibrate the frozen-CLIP view-fusion weights; metrics use the disjoint query split. R@0.5 denotes recall at 0.5 m position error. O-VIB payload includes a 20-byte representation header.
\end{minipage}

\vspace{-2mm}
\end{table}

Table~\ref{tab:real_ablation} and Figs.~\ref{fig:real_system}(c) and (d) report the held-out results. At \(k=32\), O-VIB obtains 0.293 m mean UAV error with a 0.145 KB representation, compared with 0.407 m at 10 KB for all-view CLIP retrieval. These are reductions of 28.0\% in mean error and 98.6\% in descriptor traffic. The largest prefix further lowers the mean error to 0.274 m.

\emph{Supplementary UGV extension.} For the UGV, all four views reduce the CLIP mean error from 0.622 to 0.346 m, and its mean query-to-nearest-reference distance is 0.316 m. O-VIB at \(k=32\) obtains 0.296 m mean error with a 0.145 KB representation, reducing the mean error by 14.4\% and the descriptor traffic by 98.2\% relative to the 8 KB all-view CLIP baseline. However, O-VIB has a higher UGV p90 error and a lower R@0.5 than CLIP, as regression removes reference-grid quantization for many queries but introduces several larger residuals. The indoor tests evaluate the localization accuracy on captured observations, and the latency is measured separately on the Jetson platform.

\section{Conclusion}\label{sec:conclusion}
In this paper, we have proposed an edge-assisted multi-view localization framework for UAV applications in the low-altitude economy under GPS-challenged environments. Scalable O-VIB encoding preserves localization information within constrained uplinks, VOI-guided control requests corrections when they are useful, and value-aware scheduling allocates shared edge service. On Town05, one encoder reaches 9.18 m mean error with a 0.156 KB representation at \(k=35\). Under a matched route-traffic budget, VOI-Control reaches 4.13 m mean error, compared with 5.49 m for budgeted periodic offloading, and lowers the p95 error by 31.0\%. In high-congestion replay, the shaped scheduler reduces the edge-side p95 latency for the top-10\% high-value requests from 137.7 to 32.8 ms. Indoor UAV validation reaches 0.293 m mean error with a 0.145 KB representation, while a supplementary UGV experiment shows that the same design also applies to ground platforms. These results show that the localization quality of UAVs can be improved by jointly determining what to transmit, when to request, and which request to serve. Future work will address online context adaptation and live multi-UAV closed-loop deployment.



%


\bibliographystyle{./IEEEtran}
\bibliography{ref}

\end{document}